\documentclass[11pt,letter]{article}
\usepackage{times}
\usepackage{fullpage,amsfonts,latexsym,graphicx,amssymb}
\usepackage{amsmath,amsthm,amstext,url,ifpdf}
\usepackage{ifpdf}

\usepackage{graphicx,amssymb}
\usepackage{amsmath,amstext,url}
\usepackage{ifpdf}

\def\full{1}
\def\draft{1}
\def\hylinks{1}
\def\usetitlepage{1}

\ifnum\draft=1
    
\else
    
\fi

\ifnum\full=1
  \def\nnspace{}
\else
  \def\nnspace{\vspace{-1.5ex}}
\fi

\newtheorem{theorem}{Theorem}[section]

\newtheorem{lemma}[theorem]{Lemma}
\newtheorem{corollary}[theorem]{Corollary}
\newtheorem{conjecture}[theorem]{Conjecture}
\newtheorem{observation}[theorem]{Observation}

\newtheorem{definition}[theorem]{Definition}

\newtheorem{algorithm1}[theorem]{Algorithm}
\newtheorem{protocol}[theorem]{Protocol}

\newcommand{\BBBtmp}{}

\newcommand{\BBtmp}{}

\newcounter{BBtmpC}

\newcommand{\eqdef}{\stackrel{def}{=}}

\newcommand{\etal}{{\it et al. }}

\newcommand{\Ex}{\mathbb{E}}

\newcommand{\e}{\epsilon}

\newcommand{\beq}{\begin{eqnarray*}}
\newcommand{\eeq}{\end{eqnarray*}}

\newcommand{\nucol}{\{\nu_x\}_{x\in V}}
\newcommand{\nuright}{\nu}
\newcommand{\nut}{\tilde{\nu}}
\newcommand{\nup}{\hat{\nu}}

\newcommand{\myand}{\wedge}
\newcommand{\xvec}{{\mathbf x}}
\newcommand{\bvec}{{\mathbf b}}

\newcommand{\authnote}[2]{}

\newcommand{\half}{\frac{1}{2}}

\ifnum\hylinks=1

\ifpdf
\usepackage[pdftex,colorlinks=FALSE,linkcolor=blue]{hyperref}
\else

\fi

\ifpdf
\newcommand{\sectionref}[1]{\hyperref[#1]{Section~\ref*{#1}}}
\newcommand{\theoremref}[1]{\hyperref[#1]{Theorem~\ref*{#1}}}
\newcommand{\definitionref}[1]{\hyperref[#1]{Definition~\ref*{#1}}}
\newcommand{\figureref}[1]{\hyperref[#1]{Figure~\ref*{#1}}}
\newcommand{\lemmaref}[1]{\hyperref[#1]{Lemma~\ref*{#1}}}
\newcommand{\claimref}[1]{\hyperref[#1]{Claim~\ref*{#1}}}
\newcommand{\constructionref}[1]{\hyperref[#1]{Construction~\ref*{#1}}}
\newcommand{\itemref}[1]{\hyperref[#1]{Item~\ref*{#1}}}
\newcommand{\propertyref}[1]{\hyperref[#1]{Property~\ref*{#1}}}
\newcommand{\protocolref}[1]{\hyperref[#1]{Protocol~\ref*{#1}}}
\newcommand{\algorithmref}[1]{\hyperref[#1]{Algorithm~\ref*{#1}}}
\newcommand{\equationref}[1]{\hyperref[#1]{Equation~(\ref*{#1})
}}

\else

\newcommand{\sectionref}[1]{Section~\ref{#1}}
\newcommand{\theoremref}[1]{Theorem~\ref{#1}}
\newcommand{\definitionref}[1]{Definition~\ref{#1}}
\newcommand{\figureref}[1]{Figure~\ref{#1}}
\newcommand{\lemmaref}[1]{Lemma~\ref{#1}}
\newcommand{\claimref}[1]{Claim~\ref{#1}}
\newcommand{\constructionref}[1]{Construction~\ref{#1}}
\newcommand{\itemref}[1]{Item~\ref{#1}}
\newcommand{\propertyref}[1]{Property~\ref{#1}}
\newcommand{\protocolref}[1]{Protocol~\ref{#1}}
\newcommand{\algorithmref}[1]{Algorithm~\ref{#1}}
\newcommand{\equationref}[1]{Equation~\eqref{#1}}

\fi

\newcommand{\eps}{\epsilon}

\newcommand{\one}{\mathbf{1}}

\newcommand{\Patrascu}{P\u{a}tra\c{s}cu{ }}

\begin{document}

\title{Lower Bounds on Near Neighbor Search via Metric Expansion }
\date{}

\author{
Rina Panigrahy\\
       Microsoft Research \\
       \texttt{rina@microsoft.com}
       \and
Kunal Talwar\\
       Microsoft Research \\
       \texttt{kunal@microsoft.com}
\and Udi Wieder\\
Microsoft Research \\
       \texttt{uwieder@microsoft.com}}

 \maketitle

\ifpdf
\pdfbookmark[1]{Abstract}{abs}
\fi

\begin{abstract}

In this paper we show how the complexity of performing nearest neighbor (NNS) search on a metric space is related to the expansion of the metric space.  Given a metric space we look at the graph obtained by connecting every pair of points within a certain distance $r$ . We then look at various notions of expansion in this graph relating them to the cell probe complexity of NNS for randomized and deterministic, exact and approximate algorithms. For example if the graph has node expansion  $\Phi$  then we show that any deterministic $t$-probe data structure for $n$ points must use space $S$ where  $(St/n)^t > \Phi$. We show similar results for randomized algorithms as well. These relationships can be used to derive most of the known lower bounds in the well known metric spaces such as $l_1$, $l_2$, $l_\infty$ by simply computing their expansion. In the process, we strengthen and generalize our previous results~\cite{PTW08}. Additionally, we unify the approach in~\cite{PTW08} and the communication complexity based approach. Our work reduces the problem of proving cell probe lower bounds of near neighbor search to computing the appropriate expansion parameter.

In our results, as in all previous results, the dependence on $t$ is weak; that is,  the bound drops exponentially in $t$.  We show a much stronger  (tight) time-space tradeoff for the class of \emph{dynamic} \emph{low contention} data structures. These are data structures that supports updates in the data set and that do not look up any single cell too often.
\end{abstract}

\ifnum\usetitlepage=1

\ifnum\hylinks=1

\fi

\newpage
\fi

\section{Introduction}
In the Nearest Neighbor Problem we are given a data set of  $n$ points $x_1,...,x_n$ lying in a metric space $V$. The goal is to preprocess the data set into a  data structure
such that when given a query point $y \in V$, it is possible to recover the data set point which is closest to $y$ by querying the data structure at most $t$ times. The goal is to keep both the querying time $t$  and the data structure space $m$ as small as possible.
Nearest Neighbor Search is a fundamental problem in data structures with numerous applications to web algorithms, computational biology, information retrieval, machine learning, etc. As such it has been researched extensively.

The time space tradeoff of known solutions crucially depend upon the underlying metric space. Natural metric spaces include the spaces $\mathbb{R}^d$ equipped with the $\ell_1$ or $\ell_2$ distance, but other metrics such as $\ell_\infty$, edit distance and earth mover distance may also be useful.
The known upper bounds exhibit  the `curse of dimensionality':for $d$ dimensional spaces either the space or time complexity is exponential in $d$. More efficient solutions are known  when  considering approximations e.g. \cite{IM98}, \cite{KOR98}, \cite{Indyk01}, \cite{AndoniI08}, however, in general these algorithms still demonstrate a relatively high complexity.

There is a substantial body of work on lower bounds covering various metric spaces and parameter settings; we discuss the known bounds in Section~\ref{sec:relatedwork}. Traditionally, cell probe lower bounds for data structures have been shown using communication complexity arguments~\cite{MNSW98}. \Patrascu and Thorup~\cite{PatrascuT06a} use a direct sum theorem along with the richness technique to obtain lower bounds for deterministic algorithms. Andoni, Indyk and \Patrascu~\cite{AIP06} showed randomized lower bounds using communication complexity lower bounds for Lopsided Set Disjointness. In a previous work~\cite{PTW08}, the authors used a more direct geometric argument to show lower bounds for randomized algorithms for the search version of the problem.

In this work we strengthen and significantly generalize our previous results. We give a common framework that unifies almost all known cell probe lower bounds for near neighbor search. At one extreme, it gives us the communication complexity lower bounds, and implies e.g. the result of~\cite{AIP06}. At the other extreme, we get direct data structure lower bounds leading to a strengthening to the decision problem of our results in~\cite{PTW08}. Our work in fact shows that all near neighbor lower bounds follow from basic expansion properties of the metric space. Vertex expansion translates to lower bounds for deterministic data structures. Edge expansion can be translated to lower bounds for randomized data structures, and this lets us strengthen~\cite{PTW08}.  We also identify a new (to our knowledge) graph parameter that interpolates between vertex and edge expansion, that we call {\em robust expansion}. We show that robust expansions suffices to prove NNS lower bounds. Additionally, for random inputs in highly symmetric metrics, robust expansion also translates to upper bounds in the cell probe model, that match our lower bounds for constant $t$. Finally, we present a natural conjecture regarding the complexity of approximate near neighbor search and show tight bounds for dynamic low contention data structures.

\medskip

\subsection{Basic Definitions}
The \emph{Near Neighbor Problem} is parameterized by a number $r$. As in the Nearest Neighbor Search Problem the input to the preprocessing phase is a data set of $n$ points in a metric space. Given a query point $y$ the goal is to determine whether the data set contains a point of distance at most $r$ from $y$. In the approximation version (ANNS) the preprocessing phase receives as input also an approximation ratio $c$. Given a query point $y$ the goal is to differentiate between the case where the closest data set point is of distance at most $r$ from $y$, to the case where the closest data set point is of distance at least $cr$ from $y$. Clearly a lower bound for these problems holds also for nearest neighbor search.

We prove lower bounds for a generalization we call Graphical Neighbor Search (GNS) which we define shortly. We then show that lower bounds on GNS imply  ANNS lower bounds.
In the GNS problem we are given an undirected bipartite graph $G=(U,V,E)$ where the data set comes from $U$ and queries come from $V$.  For a node $u$ the set $N(u)$ denotes its neighbors in $G$. In the preprocessing phase we are given a set of pairs $(x_1,b_1),\ldots,(x_n,b_n)$ where $x_i$ is a vertex in $U$  and $b_i\in \{0,1\}$. The goal is to build a data structure such that given a node $y\in V$, if there is a unique $i$ such that $y \in N(x_i)$ then it is possible to query the data structure $t$ times and output $b_i$. If there is no such $i$ or it is not unique any output is considered correct.

We observe that ANNS reduces to GNS when assuming a query point is at distance at most $r$ from some $x_i$ and a least $cr$ from all other $x_j$. In this case we have the nodes of $U$ and $V$ correspond to the points in the metric space, and the set of edges consists of all pairs of nodes at distance at most $r$.  A formal reduction is proven in Section~\ref{sec:applications} where we also show that average instances of ANNS translate to average instances of GNS for which our lower bounds hold. The bounds we show depend only on the expansion properties of $G$. We need the following definitions:

\begin{definition}[Vertex expansion]
Let $\mu$ be a probability measure over $U$ and $\nu$ be a probability measure over $V$. The $\delta-$vertex expansion of the graph with respect to $\mu,\nu$ is defined as
 $$\Phi_v(\delta) := \min_{A\subset V, \nu(A) \leq \delta} \frac{\mu(N(A))}{\nu(A)}.$$
The vertex-expansion $\Phi_v$ is defined as the largest $k$ such that for all $\delta \leq \frac{1}{2k}$, $\Phi_v(\delta) \geq k$.

\end{definition}

Let  $A\subset V$, $B \subset U$ and $\delta = \nu(A)$. Observe that if $E(A,B) = E(A,U)$ then $\mu(B) \geq \Phi_v(\delta)\nu(A)$. In other words $\Phi_v(\delta)$  bounds the measure of the sets that cover all the edges incident on a set of measure $\delta$. The notion of {\em robust  expansion} relaxes this by requiring $B$ to  cover at least a $\gamma$-fraction of the edges incident on $A$. This idea is captured in the definition below. For simplicity we assume that $V = U$ and that $\mu$ and $\nu$ are the uniform distribution and that $G$ is regular. A more subtle definition which takes into account other measures is presented in Section~\ref{sec:randomized}.

\begin{definition}[Robust  expansion]
$G$ has robust-expansion $\Phi_{r}(\delta,\gamma)$ if
$\forall A,B \subseteq V$ satisfying $|A| \le \delta|V|, |B| \le  \Phi(\delta,\gamma)|A|$, it is the case that $\frac{|E(A,B|}{|E(A,V)|} \le \gamma$.   Note that  $\Phi_{r}(\delta,1) = \Phi_v(\delta)$.
\end{definition}

\subsection{Our Contributions}
\subsubsection{Bounds for Deterministic Algorithms}
In this section we require that the algorithm always output the correct answer.
We show time space tradeoffs based on the vertex expansion properties of $G$. Our lower bounds are in the average case. Given a distribution $\mu$ over $U$, a data set is built by sampling $n$ data set points independently from $\mu$.

Note that in order for the problem to be interesting we must have that $N(x_i)$ and  $N(x_j)$ are likely to be disjoint. We thus have the following definition:
\begin{definition}
A distribution $\mu$ over $U$ is said to be \emph{strongly independent} for $G$ if
$$\mathop{\Pr_{x\sim \mu}}_{ z\sim \mu}\{N(x) \cap N(z) \neq \emptyset \}\leq 1/100n^2.$$

\end{definition}

Note that if $\mu$ is strongly independent and $x_1,\ldots,x_n$ are sampled independently by $\mu$ then with probability at least $0.99$  $N(x_i)\cap N(x_j) = \emptyset$ for all $i\neq j$.
In the following $m$ denotes the number of cells in the data structure and $w$ denotes the word size in bits, $t$ is the number of cell probes used by the algorithm.

\begin{theorem}\label{thm:deterministic}
For a given $G$, let $\mu, \nu$ be probability measures such that $\mu$ is  strongly independent, and the vertex expansion with respect to $\mu,\nu$ is $\Phi_v(\cdot)$. Then any deterministic algorithm solving GNS must satisfy the  following inequalities
\begin{align}
\left(\frac{mwt}{n}\right)^t &\ge \Phi_v \label{eq:det1}\\
\frac{m^ttw}{n} &\ge \Phi_v(1/m^t) \label{eq:det2}
\end{align}
\end{theorem}

These theorems, combined with known isoperimetric inequalities yield most known cell probe lower bounds for near neighbor problems, and generalize them to general expanding metrics. To see this consider for example the $d-$dimensional hypercube equipped with the Hamming distance. It is shown in \cite{PatrascuT06a}, \cite{Liu04} that any deterministic solution for ANNS with approximation $1/\e$ must satisfy $t \geq d\e^3/\log(mwd/n)$. This bound can be  slightly improved by creating the following GNS instance: Let $U$ and $V$ both equal the set of nodes of the hypercube, and let  $E = \{(u,v) ~:~|u-v|_1 \leq \e d\}$. Let $\mu$ and $\nu$ be the uniform distribution. Chernoff bounds implies that for $d=\Omega(\log n)$, $|u-v|_1 \geq 0.49d$ with overwhelming probability, so $(G,\mu)$ is a strongly independent instance. A lower bound on this instance of GNS implies a lower bound on ANNS with approximation $1/\e$.

Now we use known isoperimetric properties: Harper's theorem (see e.g. \cite{Bollobas86})  implies that there is a constant $a>1$ such that $\Phi_v \geq a^{\e^2 d}$. Plugging this in \eqref{eq:det1} we have that $t \geq d\e^2 \log a/\log(mwd/n)$.
In Section~\ref{sec:applications} we discuss how to apply these theorems in greater length.

\subsubsection{Bounds for Randomized Algorithms}
Assume that $G$ is regular. Let $x$ and $z$ be vertices drawn uniformly at random, and $y$ be a random neighbor of $x$. We say $G$ has the property of being \emph{weakly independent} if $\Pr[y\in N(z)]\leq \gamma/n$ for a small enough constant $\gamma$.
\begin{theorem}\label{thm:randomized}
There exists an absolute constant $\gamma$ such that the following holds.
Any randomized algorithm for a weakly independent instance of GNS which is correct with probability at least half (where the probability is taken over the sampling of the input and the algorithm), satisfies the following inequalities:
\begin{align}
(\frac{mwt^4}{n})^{2t} &\ge \Phi_{r}(\frac{1}{m}, \frac{\gamma}{t}) \label{eq:rand1}\\
\frac{m^{t}w}{n} &\ge \Phi_{r}(\frac{1}{m^t},\frac{\gamma}{t}) \label{eq:rand2}
\end{align}
\end{theorem}

As an example, we show in Section~\ref{sec:applications} that for the Hypercube with $E = \{(u,v) ~:~|u-v|_1 \leq (\frac{1}{2}-\e) d\}$, the robust expansion $\Phi_{r}(\frac{1}{m^t},o(1)) \geq \frac{1}{m^{t(1-4\e^2)}}$. For $d = \Omega(\log n/\e^2)$, the weak independence property is easy to verify. Plugging this into Equations ~\ref{eq:rand2}, we conclude that $m^{4t\e^2}w \ge n$ so that $m \ge (\frac{n}{w})^{\frac{1}{4t\e^2}}$. This result was previously shown by~\cite{AIP06} for slightly larger $d$.

\medskip

Our framework suggests a natural conjecture on the complexity of approximate near neighbor problems.

\begin{conjecture}
Any randomized $t$-probe datastructure for a weakly independent GNS instance must satisfy
$\frac{mw}{n} t \ge \Phi_{r}(\tfrac{1}{m}, \tfrac{1}{2t})^{\Omega(1)}$.
\end{conjecture}

We point out that for some interesting metric spaces such as the Hamming cube and Euclidean space, the known upper bound matches the lower bound in the conjecture for a wide range of parameters. We next present some evidence in support of this conjecture.

\subsubsection{An Upper Bound}

There are cases where the bounds above are known to be tight when $t=O(1)$. We show that this is no coincidence: In Section~\ref{sec:ubound} we show that if $G$ is symmetric, there is an algorithm in the cell probe model that solves random instances of GNS using space that matches the lower bound in equation (\ref{eq:rand2}) for $t=1$.

\subsubsection{Dynamic Data Structure}

In the dynamic version of the problem we want the data structure to support the operation of inserting and deleting a point in the data-set. Let $t_U$ denote the update time.
A weaker version of the conjecture is the following:
\begin{conjecture}
For any dynamic randomized $t$-probe data-structure for weakly independent GNS on $n$ points, it holds that
$t_U t \ge \Phi_r(\tfrac{1}{nt_U},\tfrac{1}{2t})^{\Omega(1)}$
\end{conjecture}

To see why this conjecture follows from the stronger one, observe that a data structure with update time $t_U$ uses space $mw \leq nt_U$ after $n$ inserts. We show that this weaker conjecture holds for a restricted family of algorithms which we call \emph{low contention}; i.e., on those where no memory location of the data structure is accessed by too many query points (see Section~\ref{sec:dynamic} for a formal definition). While this may seem like a severe limitation, we remark that known LSH data structures, and our upper bound in Section~\ref{sec:ubound}, are in fact dynamic and low contention under our definition.

We show that

\begin{theorem}\label{thm:dynamic}
For any low contention, dynamic $t$-probe datastructure for GNS on $n$ points, the update time is at least $ \Omega\left(\Phi_r(\tau,\frac{1}{4t^2})/32t^4 \right)$.
\end{theorem}

Plugging in the expansion of the hypercube, we see that for a wide range of parameters Locality Sensitive Hashing is \emph{optimal} for the class of the low contention dynamic data structures over the hypercube.

\subsection{ Related Work}
\label{sec:relatedwork}
Most previous papers are concerned with the Hamming distance over the $d$-dimensional hypercube. The cases of exact or deterministic algorithms were handled in a series of papers
\cite{CCGL99}, \cite{BOR99},\cite{Liu04}, \cite{BarkolR02}. These lower bounds hold for any polynomial space. In contrast the known upper bounds are both approximate and randomized, and with polynomial space can retrieve the output with one query. Chakrabarti and Regev~\cite{ChakrabartiR04} allow for both randomization and approximation, with polynomial space and show a tight bound for the \emph{nearest} neighbor problem. \Patrascu and Thorup\cite{PatrascuT06a} showed lower bounds on the query time of near neighbor problems with a stronger space restriction (near linear space), although their bound holds for deterministic or exact algorithms.
The metric $\ell_\infty$ is considered in an intriguing paper by Andoni \etal \cite{ACP08} who prove a lower bound for deterministic algorithms. The paper uses the richness lemma though the crux of the proof is an interesting isoperimetric bound on $\ell_\infty$ for a carefully chosen measure.

We are aware of only two papers which prove time-space lower bounds for near neighbor problems where both randomization and approximation are allowed.

Andoni, Indyk and \Patrascu~\cite{AIP06} show that for small $\e > 0$, any $O(1)$-probe algorithm for $(1+\e)$-approximate near neighbor problem must use space $n^{\Omega(\frac{1}{\e^2})}$. This bound is tight for small enough $\e > 0$ \cite{AIP06}. Panigrahy \etal \cite{PTW08} show that space  $n^{1+\Omega(\frac{1}{\e t})}$ is needed for any algorithm with $t$ queries and $\e$ approximation, for the search version of the problem. This bound is tight for constant $t$.

With the exception of~\cite{PTW08} all previous bounds were proven using  communication complexity framework~\cite{MNSW98}, and in particular the richness lemma.

\paragraph{Comparison to~\cite{PTW08}:}
While there is some overlap in the techniques between this work and~\cite{PTW08}, the current work is much more general, and stronger even for the special case (our lower bound now applies to the decision version of NNS). We show that expansion may serve as a single explanation that unifies all previous results, and also gives a simple recipe to prove lower bounds for other metrics such as $\ell_{\infty}$ and edit distance. While~\cite{PTW08} essentially contained a version of the lower bound (\ref{eq:rand1}) with the edge expansion, we are now able to additionally show (\ref{eq:rand1}). Additionally, we can use vertex expansion to show lower bounds for deterministic data structures. Moreover, we show that the randomized lower bounds hold under the much weaker notion of robust expansion.  As we discuss in Section~\ref{sec:techniques}, this strengthening is provably needed for deriving the right lower bound for the $(1+\e)$-approximation range for the Hypercube. We remark that both (\ref{eq:det2}) and (\ref{eq:rand2}) hold for communication protocol. While we do not know if (\ref{eq:det1}) and (\ref{eq:rand1}) hold for communication protocols, our proofs do shed some light on how the two approaches differ, and make clearer how the data structure is used in proving our lower bound.

\paragraph{Restricted Models:} Higher lower bounds may be achieved when considering models which are more restricted than the cell probe model. Beame and Vee \cite{BeameV02} investigate branching programs. Krauthgamer and Lee \cite{KL05} show tight upper and lower bounds for the 'black box model' where the algorithm is only allowed to query distances between points of the data set. They show that in this case the complexity of NNS is determined by the intrinsic doubling dimension of the data-set. Motwani, Naor and Panigrahy~\cite{MNP05} prove an LSH lower bound for $\ell_1$, which has recently been strengthened to the tight bound by O'Donnell, Wu and Zhou~\cite{ODonnellWZ09}.

\subsection{Notation and Preliminaries}

A data structure for Graph neighbor search is defined as follows. Given a database of $n$ points $x_1,\ldots,x_n \in U$, and $b_1,\ldots,b_n \in \{0,1\}$ the preprocessing algorithm computes a set of $t$ tables $T_1,\ldots,T_t$, where each table stores $m$ words of $w$ bits each. We often call each such word a \emph{cell} of the table. In practice there is only one table, but for notational convenience and with out loss of generality we let the data structure construct a different table for each query.

The query algorithm is specified by $t$ lookup functions $F_1,\ldots,F_t$, where $F_i$ takes in the query point $y$ and $(i-1)$ words of $w$ bits each, and outputs an integer in $[m]$, and function $F_*: V \times (2^w)^t \rightarrow \{0,1\}$. On a query $y$, the data structure looks up $c_1=T_1[F_1(y)]$, $c_2=T_2[F_2(y,c_1)],\ldots,c_t=T_t[F_t(y,c_1,\ldots,c_{t-1})]$. Finally it computes $F_*(y,c_1,\ldots,c_t)$. Note that the lookup functions, $F_i$'s and $F_*$ are fixed independent of the database, only the tables $T_1,\ldots,T_t$ can depend on $x_1,\ldots,x_t, b_1,\ldots,b_t$. We say the algorithm is \emph{non adaptive} if the lookup functions are independent of the content of the tables, i.e. of the $c$ values.

\subsection{Overview of Techniques}
\label{sec:techniques}
The core idea behind our approach is quite simple. We demonstrate it by showing a simple argument that the vertex expansion of $G$ provides a lower bound on the space of $1$-probe data structures for deterministic algorithms. By the definition of vertex expansion, every set of $|V|/\Phi_v$ nodes is incident to at least half of the nodes of $G$. Let $L$ be a uniformly random sample of a $1/\Phi_v$ fraction of the cells of the table $T$, and let $Q$ be the set of nodes in $V$ for which the algorithm probes a cell in $L$. Clearly $Q$ is expected to contain a $1/\Phi_v$ fraction of the nodes in $G$. Now consider a sample data set $(x_1,b_1),\ldots,(x_n,b_n)$ where $x_1,\ldots,x_n$ are randomly sampled nodes in the graph and $b_1,\ldots,b_n$ are random bits. With overwhelming probability at least a quarter of the $x_i$'s have a neighbor in the set $Q$, and thus the random bits associated with these points should be retrievable from the contents of $L$ alone. We conclude that the total number of bits in $L$ is at least $n/4$ and thus the space of the data structure is at least $n\Phi_v/4$ bits.

This basic sampling approach for $1$-probe data structures can be extended to $t$-probe data structures in two different ways.

\medskip\noindent{\em Cell Sampling:} Here we sample a $\Phi_v^{-\frac{1}{t}}$ fraction of the cells in each table. Thus a $1/\Phi_v$ fraction of $V$ is expected to access only the sampled cells. This immediately gives bound \eqref{eq:det1} for non-adaptive algorithms.

\medskip\noindent{\em Path Sampling:} Here we sample a {\em path} as follows: we pick a cell randomly from the first table so that a $\frac{1}{m}$ fraction of the vertices $Q_1$ lookup this cell. Then we sample a cell from the second table in such a way  that a $\frac{1}{m}$ fraction of $Q_1$ looks up this cell in the second read, and so on. This immediately leads to the lower bound in \eqref{eq:det2} for non-adaptive algorithms.

We remark that the path sampling approach actually leads to communication complexity lower bounds for the 2-player version of the problem where Alice has the query point and Bob has the database. Any $t$-probe data structure with $m$ cells of $w$ words each implies the existence of a $t$-round communication protocol where Alice sends $\log m$ bits, and Bob sends $w$ bits, in each round. A communication protocol has more freedom however; unlike in a data structure, where the same table $T_2$ is used to answer any second query, in a communication protocol, the message Bob sends in the second round may depend not just on the second message from Alice, but also on the first. Path sampling can be immediately translated to a ``transcript sampling" technique and thus gives lower bounds for communication protocols. There is no similarly obvious translation for cell sampling.

\medskip

We can extend these ideas and provide lower bounds for adaptive algorithms by observing the following two facts. Firstly, for a fixed data structure, the probability over a random data set that the data structure succeeds is exponentially small in $n$. On the other hand, the number of bits read by the sampling procedures above is sublinear, thus the number of all possible non-adaptive algorithms is sub exponential. Informally, this allows us to do a union bound over all possible values of the bits read.

In randomized algorithms not all points in $N(x)$ are good query points for $x$. In particular, the specific query point that queries the cells that are sampled may be a point on which the algorithm errs. The notion of \emph{shattering} plays a major role in extending the bound for this case: Given any fixed partitioning $A_1,\ldots,A_m$ of $V$ such that each set is of cardinality $O(|V|/m)$, a randomly chosen $x$ has (with high probability) the property that $\max_i |N(x)\cap A_i|$ is at most $|N(x)|/K$, for a $K$ that depends on the edge expansion. In other words, $N(x)$ is {\em shattered} by the partitioning. Given that the lookup algorithm is correct for a large fraction of $N(x)$, shattering suffices to show that the algorithm still gives the right answer for a \emph{majority} of the points in $N(x)$ which can be looked up from the cell sample (or the path sample).

In order to prove lower bounds for randomized adaptive algorithms we need to combine the ideas outlined in the two previous paragraphs, which requires more work. Intuitively, for every $x$ such the $N(x)$ is shattered, and for any fixed subset $N'(x)$ on which the algorithm succeeds, the sampling is very likely to recover the correct answer. Moreover, for every collection of bits read, almost all points shatter. While it would be tempting to use a union bound at this point, that does not quite work. Informally, there are dependencies everywhere: the part of $N(x_i)$ that the algorithm gets right depends on all the other $x_j$'s, the bits that are read depend on the sampled cells, etc. The proof  carefully defines a notion of shattering that depends only on the $x$'s and not on $N'(x)$'s and argues (over the randomness in picking the $x_i$'s) that most points get shattered. Separately, we argue that for a point that gets shattered, and for any fixed $N'(x)$, the majority answer is correct with high probability (over the sampling procedure alone).

\medskip

The notion of edge expansion does not quite suffice: for the hypercube when $r=(\frac 1 2 -\eps)d$, for fixed partitioning $A_1,\ldots,A_m$ of $V$ into cells of size $|V|/m$, the largest $|N(x)\cap A_i|/|N(x)|$ is likely to be quite large ($ \approx \frac{1}{m^{\eps}}$), whereas we would need it to be $\frac{1}{m^{\eps^2}}$ to get the correct bound. The definition of robust expansion $\Phi_{r}$ comes to our rescue here. We can show that while the largest $|N(x)\cap A_i|$ is usually large, the large pieces account for a very small fraction of $N(x)$. In fact, after removing a vanishingly small fraction of $|N(x)|$, every other piece is only about $\frac{1}{m^{\eps^2}}$. We show that our lower bound proofs are robust enough to handle this weaker notion of shattering.

While our techniques do not improve the dependency on the query time $t$, they overcome some of the inherent obstacles in the richness method, so for instance, strengthening the isoperimetric bound of $\ell_\infty$ would imply that the bound in \cite{ACP08} extends to randomized algorithms as well.

\section{Deterministic Lower bounds}
In this section we prove Theorem~\ref{thm:deterministic}. The analysis of deterministic algorithms involves node expansion and does not require shattering. It allows us to demonstrate the techniques of cell sampling and path sampling in a simple setting.

\subsection{Cell Sampling}
The following theorem is a restatement of  inequality~\eqref{eq:det1}
\begin{theorem}\label{thm:det1}
Let $\mu$ be such that $(G,\mu)$ is a strongly independent instance, and $\Phi_v$ be the vertex expansion with respect to $\mu,\nu$. Then any deterministic algorithm solving GNS must satisfy
$(\frac{8mwt}{n})^t \ge \Phi_v$.
\end{theorem}
\begin{proof}
Recall that $T_i$ represents a table with $m$ cells, from which the $i'th$ query reads, and $F_i:V \rightarrow T_i$ denotes the $i$'th lookup function.
We will state a procedure that obtains a set of at most $tm/\Phi_v^{1/t}$  cells  so that at least a $1/\Phi_v$ fraction of the query points only access these cells. We call this procedure \emph{cell sampling}. Note that here this procedure is entirely deterministic. A probabilistic variant is used in the next section.

{\it Cell sampling procedure:} The cells are obtained iteratively in $t$ phases, each phase corresponds to a query of the table. In each phase at most $m/\Phi_v^{1/t}$ cells are chosen. The first lookup function $F_1$ induces a partition over $V$. The set $L_1$ is chosen to be the $m/\Phi_v^{1/t}$ cells in $T_1$ that maximize $\nu(F_1^{-1}(L_1))$. In other words, the first lookup function partitions $V$ according to its image in $T_1$. We choose $L_1$ to be the cells corresponding to the $m/\Phi_v^{1/t}$ largest partitions, as measured by $\nu$. Set $Q_1\subset V$ to be those vertices; i.e. $F_1^{-1}(L_1)$. The selection process continues iteratively in a similar manner. Let $L_i$ denote the set of cells obtained in the $i$'th phase and let $Q_i \subseteq V$ denote the set of vertices that (if given as a query) access only cells in $L_1,\ldots,L_i$ In the $(i+1)$'th phase we consider $F_{i+1}$ and set $L_{i+1}$ to be the $m/\Phi_v^{1/t}$ cells with highest measure, where we restrict $\nu$ to $Q_{i}$. In other words, when measuring $F_i^{-1}(L_{i+1})$ we assign a measure of $0$ for vertices outside $Q_{i}$. It is easy to inductively argue that $\nu(Q_i) \geq \Phi_v^{-i/t}$, so that $\nu(Q_t) \geq 1/\Phi_v$ and thus $\mu(N(Q_t)) \geq \half$.

Intuitively, the set of cells in $L_1,\ldots,L_t$ encode half of the $b_i$ bits and therefore must contain $\Omega(n)$ bits, which would imply the lower bound. Of course, $L_1,\ldots,L_t$  depends upon the content of the table which depends upon the points in the data set. These dependencies can be handled by using a union bound over all the possible values of $L_t$:
To see this, fix the values written in the cells $L_1,\ldots,L_t$ to some string $\omega$ and sample the $n$ data set points from $U$ independently according to $\mu$. Let $A_\omega$ denote the event that when the value of the cells $L_1,\ldots,L_t$ is $\omega$, an algorithm reading $\omega$ succeeds in guessing the $b_i$ bits for the data set points that fall in $N(Q_t)$.  Note that $Q_t$ depends only upon $\omega$ and that since the procedure of obtaining $L_t$ is deterministic, the locations of the cells obtained also depends only on $\omega$. Vertex expansion implies that $\mu(N(Q_t)) \geq \half$. Also, since $\mu$ is strongly independent and the algorithm is assumed to be correct, $\Pr[\cup_\omega A_\omega] \geq \half$. By Chernoff's bound, the  probability that less than $n/8$ points fall in $N(Q_t)$ is at most $2^{-n/8}$. Note that the $b_i$'s are chosen independently, therefore, for a fixed table, if $n/8$ points indeed fall in $N(Q_t)$, then the probability that the sampled $b_i$'s match the output of the algorithms is $2^{-n/8}$. We conclude that $\Pr[A_\omega] \leq 2^{-n/8}+2^{-n/8}$. Now, let $K = 1/\Phi_v^{1/t}$. There are $2^{Kmtw}$ ways of choosing $\omega$, so we must have $2^{1-n/8} 2^{Kmtw}\geq \half$. We conclude that $Kmtw \geq n/8$ which implies the theorem.
\end{proof}

\subsection{Path Sampling}
We now prove Inequality \eqref{eq:det2}.
\begin{theorem}
Let $\mu$ be  strongly independent, and $\Phi_v$ be the vertex expansion with respect to $\mu,\nu$, then any data structure with a deterministic querying algorithm must satisfy
$\frac{9m^twt}{n} \ge \Phi_v(1/m^t)$.
\end{theorem}

\begin{proof}
The proof is similar to that of Theorem~\ref{thm:det1} with a different choice of parameters. We present it here separately because the two approaches diverge in the next section when we deal with the randomized case.  We use the cell sampling technique to select a set of cells from the tables, only this time, each phase we select \emph{a single} cell from the table (as opposed to selecting $m\Phi_v^{-1/t}$ cells in Theorem~\ref{thm:det1}). We call the approach of sampling a single cell from each table \emph{path sampling} because we sample a single possible ``query path'' along the $t$ tables. We also observe that lower bounds based on this approach imply communication complexity lower bounds.

Now the contents of $L_1,\ldots,L_t$ are $tw$ bits, and $\nu(Q_t) \geq m^{-t}$ so that $\mu(N(Q_t)) \geq \Phi_v(m^{-t})m^{-t}$. When fixing the bits of $L_t$ to be the string $\omega$, the expected number of data set points that fall in $N(Q_t)$ is at least $n\Phi_v(m^{-t})m^{-t}$. Define $A_\omega$ as before and recall that we have $\Pr[\cup A_\omega] \geq \half$. Let $Z_\omega$ be the number of data set points falling in $N(Q_t)$. Chernoff's bound implies that $\Pr[Z_\omega \leq \half n\Phi_v(m^{-t})m^{-t}] \leq 2^{-\Phi_v(m^{-t})m^{-t}/8}$. Since the string $\omega$ now encodes $Z_\omega$ random bits, $\Pr[A_\omega] \le 2^{-\Phi_v(m^{-t})m^{-t}/8} + 2^{-\Phi_v(m^{-t})m^{-t}/2}$. There are $2^{tw}$ ways of choosing $\omega$, so we have $2^{tw}\cdot 2^{-\Phi_v(m^{-t})m^{-t}/8} \geq \half$ which implies the theorem.
\end{proof}
\section{Randomized Lower bounds}
\label{sec:randomized}

\subsection{Preliminaries}

To prove lower bounds for randomized data structure, we will use Yao's minimax theorem, and instead show a distribution over instances such that for some constant $\delta>0$, any deterministic $t$-probe data structure that succeeds with probability $(1-\delta)$ needs large space.

We consider the following randomized version of the {\em Graph Neighbor Search (GNS)} problem on a bipartite graph $G=(U,V,E)$. We are given a set of $n$ tuples $(x_1,b_1),\ldots,(x_n,b_n)$, where $x_i \in U$ and $b_i \in\{0,1\}$ to preprocess into a data structure. Then given a query $y \in V$, the query algorithm makes $t$ probes into the data structure, and is expected to return $b_i$ if $x_i$ is the unique neighbor in $G$ of $y$ in $\{x_1,\ldots,x_n\}$ (if there is no unique neighbor, any output is considered valid).

Let $G=(U,V,E)$ be a bipartite graph and let $e$ be a probability distribution over $E$. Let $\mu(u) = e(u,V) = \sum_{v\in V} e(u,v)$ be the induced distribution on $U$, and let $\nu(v) = e(U,v)$ be the induced distribution on $V$. For $x \in U$, we denote by $\nu_x$ the conditional distribution of the endpoints in $V$ of edges incident on $u$, i.e. $\nu_x(y) = e(x,y)/e(x,V)$.

Suppose we have a graph $G=(U,V,E)$, and the distribution $e$ on $E$. Then $(G,e)$ define a distribution over instances of GNS as follows. We select $n$ points $x_1,\ldots,x_n$ independently from the distribution $\mu$ uniformly at random and pick $b_1,\ldots,b_n$ independently and uniformly from $\{0,1\}$. This defines the database distribution. To generate the query, we pick an $i \in [n]$ uniformly at random, and sample $y$ independently from $\nu_{x_i}$.

We say the tuple $(G,e)$ satisfies $\gamma$-{\em weak independence (WI)} if $\Pr_{x,z \sim \mu, y \sim \nu_x}[(y,z) \in E] \leq \frac{\gamma}{n}$. In other words, WI ensures that with probability $(1-\gamma)$, for the instance generated as above, $x$ is indeed the unique neighbor in $G$ of $y$ in $\{x_1,\ldots,x_n\}$.

We next define the notion of expansion that we use. Recall that the vertex expansion of a set $A\subseteq V$ in an unweighted graph $G$ is the ratio $\frac{|B|}{|A|}$, where $B=N(A)$ is the smallest set such that all edges incident on $A$ are captured by $B$, i.e. $|E(B,A)| = |E(U,A)|$. A relaxation of this definition, which we call $\gamma$-robust expansion, is the ratio $\frac{|B|}{|A|}$ where $B$ is now the smallest set that captures a $\gamma$ fraction of the edges incident on $A$, i.e. $e(B,A) \geq \gamma e(U,A)$. The following definition generalizes this notion to weighted bipartite graphs.

\begin{definition}
[Robust Expansion]
The $\gamma$-robust expansion of a set $A \subseteq V$ is defined as $\phi_{r}(A,\gamma) \eqdef \min_{B \subseteq U: e(B,A) \geq \gamma e(U,A)} \mu(B)/\nu(A)$.

Let $w$ be an auxiliary weight function on $U$ with $\sum_{u \in U} w(u) = 1$. The $\gamma$-robust expansion with respect to $w$ is defined as $\phi^w_{r}(A,\gamma) \eqdef \min_{B \subseteq U: e(B,A) \geq \gamma e(U,A)} w(B)/\nu(A)$.

We say that $(G,e)$ has $(\beta,\gamma)$-robust expansion $\phi^w_{r} = \phi^w_{r}(\beta,\gamma)$ at least $K$ if for every subset $A \subseteq V$ such that $\nu(A) \leq \beta$, we have $\phi^w_{r}(A,\gamma) \geq K$.
\end{definition}

For intuition, consider the setting where $G =(U,V,E)$ is derived naturally from an undirected graph $H = (V_H,E_H)$ by making two copies of $V_H$ and for each edge $(u,v) \in E_H$, placing the edges $(u_1,v_2)$ and $(v_1,u_2)$. Formally, $U_G = V_H \times \{1\},  V_G = V_H\times \{2\}$, and $E_H = \{((u,1),(v,2)) \in U_G \times V_G : (u,v) \in E_H\}$. Then for any set $A \subseteq V$, we have $\phi^w_{r}(A,1) = w(N(A))/\nu(A)$, which is the vertex expansion of $A$ in $H$ under $\nu$, for $w=\nu$. Similarly, if a set $A$ has conductance $e(A,A^c)/e(A,V_H)$ at most $1-\gamma$, then $e(A,A) \geq \gamma e(A,V)$ so that $\phi^w_{r}(A,\gamma) \leq w(A)/\nu(A)$. A similar correspondence holds for directed graphs.

We next give some more definitions.
\begin{definition}
[$\beta$-sparse] A collection $A_1,\ldots,A_k$ of disjoint subsets of $V$ is said to be $\beta$-sparse with respect to $(G,e)$ if $\max_i \nu(A_i) \leq \beta$.
\end{definition}

We now recall the notion of strong shattering.

\begin{definition}[Strong Shattering]
Given $(G,e)$ and a collection $A_1,\ldots,A_k$ of disjoint subsets of $V$, we say the collection $\{A_i\}_i$ {\em $K$-strongly shatters} a point $x \in V$ if $\max_i \nu_x(A_i) \leq \frac{1}{K}$.
\end{definition}

We shall in fact show our lower bounds using a weaker notion of shattering, which allows a small probability mass from $\nu_x$ to be in $A_i$'s with $\nu_x$ measure larger than $\frac{1}{K}$. For a real number $a$, let $(a)^+ \eqdef \max(a,0)$ denote the positive part of $x$. Note that strong shattering says that each of the $\nu_x(A_i)$'s is at most $\frac{1}{K}$ so that $\sum_{i} (\nu_x(A_i) - \frac{1}{K})^+$ is zero. We relax this condition.

\begin{definition}
[Weak Shattering]
Given $(G,e)$ and a collection $A_1,\ldots,A_k$ of disjoint subsets of $V$, we say the collection $\{A_i\}_i$ {\em $(K,\gamma)$-weakly shatters} a point $x\in V$ if $\sum_{i} (\nu_x(A_i) - \frac{1}{K})^+ \leq \gamma\nu(\cup_i A_i)$.

\end{definition}

\begin{definition}[$(K,\beta,\gamma)$-weak shattering (WS) property]
We say a tuple $(G,e)$ satisfies the $(K,\beta,\gamma)$-weak shattering (WS) property if for every $\beta$-sparse collection $A_1,\ldots,A_k$ of disjoint subsets of $V$,
$$\Pr_{x \sim \mu}[A_1,\ldots,A_k \mbox{ $(K,\gamma)$-weakly shatters } x ] \geq 1-\gamma.$$
\end{definition}

We record the following implication of weak shattering.
\begin{observation}
\label{obs:Kbeta-transfer}
If $(G,e)$ satisfies $(K,\beta,\gamma)$-weak shattering property, then it also satisfies $(K/\lceil\frac{\beta'}{\beta}\rceil,\beta',\gamma)$-weak shattering for any $\beta'>0$.
\end{observation}
\begin{proof}
For $\beta' < \beta$, there is nothing to prove since every $\beta'$-sparse collection is also $\beta$-sparse. For $\beta'>\beta$, we can arbitrarily break each set $A_i$ into $s=\lceil\frac{\beta'}{\beta}\rceil$ pieces to derive a $\beta$-sparse collection; the shattering follows using the identity $(\sum_{i=1}^s a_i)^+ \geq \sum_{i=1}^s (a_i)^+$.
\end{proof}

Also observe that
\begin{observation}
\label{obs:shatternut}
Let $A_1,\ldots,A_k$ be a collection of disjoint subsets of $V$. Then for any $x \in V$, there is a measure $\nup_x$ such that (a) $\nu_x$ dominates $\nup_x$, i.e. $\nup_x(A) \leq \nu_x(A)$ for all $A \subseteq V$, (b) $\nup_x(A_i) \leq \frac{1}{K}$ for all $i \in [k]$, and (c) If $x$ is $(K,\gamma)$-weakly shattered, then $ \nup_x(\cup A_i) \geq \nu_x(\cup A_i) - \gamma \nu(\cup A_i)$.
\end{observation}
Note that $\nup$ is not necessarily a probability measure. Intuitively, $\nup$ is a part of the measure $\nu$ that gets shattered. Such a measure can be obtained by shaving the mass on $y$ that fall in clusters with large $\nu_x$ mass.
\begin{proof}
For each $A_i$ with $\nu_x(A_i) \geq \frac{1}{K}$, we set $\nup_x(y) = \frac{\frac{1}{K}}{\nu_x(A_i)} \nu_x(y)$ for each $y \in A_i$. $\nup_x(y)$ is set $\nu_x(y)$ for the remaining $A_i$'s. The dominance is immediate, and the small loss property follows from the fact that $\nup_x(A_i) = \frac{1}{K}$ for every $A_i$ of the first type so that $\nu_x(A_i) - \nup_x(A_i) = (\nu_x(A_i) - \frac{1}{K})^+$.
\end{proof}

Finally, we shall use the following simple information-theoretic lemma:
\begin{lemma}
Let $\delta < \frac{1}{4}$, and let $Enc: \{0,1\}^n \rightarrow \{0,1\}^N$ and $Dec: \{0,1\}^N \rightarrow \{0,1\}^n$ be functions such that $|b - Dec(Enc(b))|_1 < \delta n$ with probability at least $\frac{1}{2}$ when $b$ is drawn at random. Then there exists a constant $r = r(\delta)>0$ such that $N \geq rn$, where $\lim_{\delta \rightarrow 0} r(\delta) = 1$.
\end{lemma}
\begin{proof}
Let $C \subseteq \{0,1\}^n$ be a binary error correcting code with positive rate and minimum distance $2\delta < \frac{1}{2}$. Then for a random $v \in \{0,1\}^n$, $\hat{C_v} = \{b \in v+C : |Dec(Enc(b)) - b|_1 < \delta n\}$ has expected size $2^{rn}$ for an $r=r(\delta)>0$. Since the minimum distance of $\hat{C_v}$ is $2\delta n$, the values $Enc(b) : b \in \hat{C_v}$ are all distinct, leading to the claim bound.
\end{proof}

The lemma can be extended to the setting where the encoder and the decoder share some randomness.
\begin{corollary}
\label{cor:nocompression}
Let $\delta < \frac{1}{4}$, and let $Enc: \{0,1\}^n \times \{0,1\}^R \rightarrow \{0,1\}^N$ and $Dec: \{0,1\}^N \times \{0,1\}^R \rightarrow \{0,1\}^n$ be functions such that $\Ex_{b,z}[|b - Dec(Enc(b,z),z)|_1] \leq (\delta/2) n$. Then there exists a constant $r = r(\delta)>0$ such that $N \geq rn$, where $\lim_{\delta \rightarrow 0} r(\delta) = 1$.
\end{corollary}
\begin{proof}
There must exist a $z$ such that $\Ex_{b}[|b - Dec(Enc(b,z),z)|_1] \leq (\delta/2) n$. By Markov's inequality, $\Pr_b[|b - Dec(Enc(b,z),z)|_1 > \delta n] \leq \frac{1}{2}$. The claim follows.
\end{proof}

\subsection{Main Result}
The main result of this section is the following.

\begin{theorem}\label{thm:randomizedfull}
There exists an absolute constant $\gamma$ such that the following holds.
Let $(G,e)$ satisfy $\gamma$-weak independence (WI). Then for $t \leq n^{\frac{1}{4}}$, any deterministic $t$-probe data structure for the distribution over GNS instances defined by $(G,e)$ that succeeds with probability $(1-\gamma)$ must satisfy
\begin{align}
(\frac{mwt^4}{n})^{2t} &\ge \Phi^w_{r}(\frac{1}{m}, \frac{\gamma}{t})\\
\frac{m^{t}w}{n} &\ge \Phi^w_{r}(\frac{1}{m^t},\frac{\gamma}{t})
\end{align}
where $w$ is an arbitrary auxiliary weight function, and $t$ is $o(n^{\frac{1}{4}})$.
\end{theorem}

The theorem will follow from Lemma~\ref{lem:exptoshatweak}, and Theorems~\ref{thm:randcellsample} and~\ref{thm:randpathsample} that we prove next.

\subsection{Expansion to Shattering}
\begin{lemma}[Expansion implies shattering]
\label{lem:exptoshatweak}
Let $A_1,\ldots,A_{k}$ be a $\beta$-sparse collection of disjoint
subsets of $V$. Then
\[
\Pr_{x \sim \mu}[x \;\;\;{\mbox is }\;\;\;(K,\gamma){\mbox-weakly
shattered }] \geq (1-\gamma)
\]
for $K=\Phi_{r}(\beta,\frac{\gamma^2}{4}) \gamma^3/16$
\end{lemma}
\begin{proof}
We will show that if $\Pr_{x \sim \mu}[x \;\;\;{\mbox is }\;\;\;(K,\gamma){\mbox-weakly
shattered}] \leq (1-\gamma)$, then one of the $A_i$'s does not expand enough, thus deriving a contradiction.

Let $\eta= \nu(\cup A_i)$. Observe that $\Ex_{x \sim \mu} [\nu_x(\cup A_i)] = \nu(\cup A_i) = \eta$. Thus by Markov's inequality, $\Pr_{x \sim \mu}[\nu_x(\cup A_i) \geq \frac{2\eta}{\gamma}] \leq \frac{\gamma}{2}$.

Suppose that $\Pr_{x \sim \mu}[x \;\;\;{\mbox is }\;\;\;(K,\gamma){\mbox-weakly
shattered}] \leq (1-\gamma)$. Thus for at least $\frac{\gamma}{2}$ fraction of $x$'s (drawn from $\mu$), $\nu_x(\cup A_i) \leq \frac{2\eta}{\gamma}$ and yet $x$ is not $(K,\gamma)$-weakly shattered. Let $B$ be the set of such $x$'s. In other words, the set $B$ satisfies
\begin{itemize}
\item $\mu(B) \geq \frac{\gamma}{2}$.
\item For each $x\in B$, $\nu_x(\cup A_i) \leq \frac{2\eta}{\gamma}$.
\item For each $x\in B$, $\sum_i (\nu_x(A_i) - \frac{1}{K})^+ \geq \gamma \eta$.
\end{itemize}

Construct an weighted graph $H_0$ on $B \times [k]$, where we put an edge between $(x,i)$ with weight $\mu(x)\nu_x(A_i) = e(x,A_i)$ if $\nu_x(A_i) \geq \frac{1}{K}$. Thus $e_{H_0}(x,i) \leq e(x,A_i)$ so that $e_{H_0}(B,i) \leq \nu(A_i)$.

The (unweighted) degree of each node $x\in B$  in $H_0$ is at most $\frac{\nu_x(\cup A_i)}{\frac{1}{K}} \leq 2\eta K/\gamma$, since each edge incident on $x$ in $H_0$ contributes at least $\frac{1}{K}$ to $\nu_x(\cup A_i)$. Moreover, by the properties of $B$ above the total edge weight in $H_0$  $e_{H_0}(B,[k])$ is at least $(\frac{\gamma}{2}) (\gamma \eta)  = \frac{\gamma^2\eta}{2}$.

Let $H_1$ be the subgraph of $H_0$ formed by deleting all nodes $i \in [k]$ such that the total edge weight $e_{H_0}(B,i)$ incident on $i$ is at most $\frac{\gamma^2}{4} \nu(A_i)$. The total edge weight deleted in the process is at most $\frac{\gamma^2}{4} \eta$ so that the total edge weight in $H_1$  $e_{H_1}(B,[k])$ is at least $\frac{\gamma^2\eta}{4}$.

Let $R \subseteq [k]$ be the set of nodes on the right surviving in $H_1$. Since each node in $B$ has unweighted degree at most $2\eta K/\gamma$ in $H_1$, $\sum_{i \in R} w(N_{H_1}(i)) \leq (2\eta K/\gamma) \sum_{x \in B} w(x) \leq 2 \eta K/\gamma$. On the other hand, $\sum_{i \in R} \nu(A_i)$ is at least the total weight $e_{H_1}(B,R)$ of edges in $H_1$ which is lower bounded by $\frac{\gamma^2 \eta}{4}$. It follows by an averaging argument that there is a node $i^* \in S$ such that
$$ \frac{w(N_{H_1}(i^*)}{\nu(A_{i^*})} \geq \frac{\sum_{i \in S} w(N_{H_1}(i))}{\sum_{i\in S} \nu(A_i)} \leq \frac{2\eta K/\gamma}{ \gamma^2 \eta/4} = \frac{8K}{\gamma^3}.$$

Let $B_{i^*} = N_{H_1}(i^*)$. Since $i^* \in R$, we have $e(B_{i^*},A_{i^*}) \geq \frac{\gamma^2}{4} \nu(A_{i^*})$. Since $\nu(A_{i^*}) \leq \beta$ by assumption, $$\Phi_{r}(\beta,\frac{\gamma^2}{4}) \geq \Phi_{r}(A_{i^*},\frac{\gamma^2}{4}) = \frac{w(B_{i^*})}{\nu(A_{i^*})} \geq \frac{8K}{\gamma^3},$$ which contradicts the definition of $K$.
\end{proof}

\subsection{Path Sampling}

In this section, we show the following theorem
\begin{theorem}
\label{thm:randcellsample}
There exists a constant $\gamma>0$ such that the following holds.
Let $(G,e)$ satisfy $\gamma$-weak independence (WI) and $(K,\frac{1}{m^t},\frac{\gamma}{t})$-weak shattering property. Then for $t \leq n^{\frac{1}{4}}$, any deterministic $t$-probe data structure for the distribution over GNS instances defined by $(G,e)$ that succeeds with probability $(1-\gamma)$ must use space $m$ at least $\Omega((nK^2/w)^{\frac{1}{t}})$.
\end{theorem}
{\bf Proof Sketch:} Suppose that a datastructure with $m < (\gamma^6nK/4 t^4w)^{\frac{1}{t}}$ exists that succeeds with probability $(1-\gamma^3)$. We use it construct a randomized encoding and decoding algorithm for a random binary vector $b \in \{0,1\}^n$. We sample $x_1,\ldots,x_n$ from $\mu$, and build the data structure for the database $(x_1,b_1),\ldots,(x_n,b_n)$ to get $t$ tables $T_1,\ldots,T_t$ where each $T_i$ contains $m$ cells of $w$ bits each.

We show how to sample $s$ cells from each table, for a suitable $s$ and let $Enc(b,z)$ be the contents of those cells (where $z$ is used as the randomness to pick the $x_i$'s and in the sampling process). The decoding algorithm essentially takes the majority answer in $\nu_{x_i}$, restricted to queries that the data structure can answer based on the sampled cells, as its guess for $b_i$. The success of the data structure, and the WI property, imply that the answer on $\nu_{x_i}$ is equal to $b_i$ with probability $(1-\gamma)$, for most $x_i$'s. We show that the weak shattering property is sufficient to guarantee (using Chernoff bounds) that the majority answer on the restriction of $\nu_{x_i}$ is still equal to $b_i$ with high probability. For suitably small $\gamma$, this violates corollary~\ref{cor:nocompression}.

Intuitively, the $t$ lookup functions break $G$ into $m^t$ pieces. We could sample $s$ of these pieces. If all $x_i$'s were strongly shattered, each piece has little influence on measure in $\nu_{x_i}$ that can be looked up from the sample. For large enough $s$, Chernoff bounds would then imply that the restricted measure has a large probability of answering $b_i$ as well, completing the proof.

The proof below, while following the above intuition, is made complicated by several factors. The lookup functions are adaptive so that the $m^t$ pieces that $V$ breaks into depends on the table contents. Thus the shattering itself depends on the table contents sampled, making a one-shot sampling argument untenable. We instead give an inductive proof, that handles these dependencies. The weaker shattering assumption forces us to slightly change the decoding algorithm, to take a majority under a modified measure. Additionally, the pieces may be of different sizes, and we need to break large pieces to ensure sparseness.

We are now ready to present a detailed proof.

\begin{proof}
We assume the contrary so that for $m < (\gamma^6nK/4 t^4w)^{\frac{1}{t}}$, there is $t$-probe space $m$ data structure that succeeds with probability $1-\gamma^4$ on the distribution defined by $(G,e)$. We use this data structure to construct functions $Enc$ and $Dec$ violating corollary~\ref{cor:nocompression}. We use the auxiliary input $z$ as shared randomness between $Enc$ and $Dec$ throughout this proof.

We first sample points $x_1,\ldots,x_n$ from $\mu$ and use $(x_1,b_1),\ldots,(x_n,b_n)$ as the database. Note that when $b$ and $z$ are random, $(x_i,b_i)$ are distributed according to $(G,e)$, and hence for a random $y$ drawn from the appropriate distribution, the data structure will return the right answer for $y$ with probability $1-\gamma^3$. By Markov's inequality with probability $(1-\gamma)$, it is the case that except for $(1-\gamma)$ of the $i$'s, the data structure answers correctly with probability $(1-\gamma)$ when $y$ is drawn from $\nu_{x_i}$. By WI, except with probability $\gamma$, $y$ has a unique neighbor in $x_1,\ldots,x_n$ so that the correct answer is $b_i$. Thus except for a $\gamma$ fraction of the $i$'s, with probability $(1-2\gamma)$, the data structure returns $b_i$ on a random $y$ from $\nu_{x_i}$. Thus the tables $T=T_1,\ldots,T_t$ are a valid, albeit long, encoding that can be decoded appropriately by taking the majority answer on $\nu_{x_i}$ as a guess for $b_i$. In the rest of the proof, we argue that in fact a random sample of the tables suffices.

\medskip\noindent{\em Sampling Procedure:}
We will sample $s$ ``paths" in the table, where $s$ is set to $s\eqdef \frac{\gamma^2 n}{4tw}$. Further we assume that $w \geq 2\log m$.

Let $F_1 = F_1(y)$, $F_2 = F_2(y,\alpha_1)$, $F_3=F_3(y,\alpha_1,\alpha_2)$  etc. be adaptive lookup functions that the data structure uses.  The $j$th path consists of $t$ cells $\Lambda_{j1},\ldots,\Lambda_{jt}$, one from each of the $t$ tables sampled sequentially. We will also get a telescoping sequence of subsets $V \supseteq A^{(1)} \supseteq, \ldots, A^{(t)}$ where $A^{(k)}$ denotes the set of queries that access the sampled cells at locations $\Lambda_{j1},\ldots, \Lambda_{j(k-1)}$ in the first $k-1$ tables, for some $j \in [s]$. Observe that the cells accessed in $T_k$ depends on the contents of the cells accessed in the previous table.

We first describe how to sample a single path $\Lambda_{11},\ldots,\Lambda_{1t}$. To sample from the first table we look at the partition of $V$ into $m$ parts induced by the value of $F_1(y)$ over all $y \in V$. Let $A^{(10)}_1,\ldots, A^{(10)}_{2m}$ be a $\frac{1}{m}$-sparse partition of $V$ that refines $\{F^{-1}(l): l \in [m]\}$. Such a partitioning can be obtained by starting with $\{F^{-1}(l): l \in [m]\}$ and repeatedly splitting parts larger than $\frac{1}{m}$ into smaller pieces. This splitting can be done arbitrarily, and results in a $\frac{1}{m}$-sparse partitioning containing at most $2m$ parts; we pad this with empty parts to get exactly $2m$ sets $A^{(10)}_1,\ldots, A^{(10)}_{2m}$. This corresponds to a table $T_1$ of size $2m$ where the cells corresponding the partitions that were split are replicated appropriately. The first cell of the path is simply obtained by picking a random index $\Lambda_{11}$ into this table. Let $A^{(11)}\eqdef A^{(10)}_{\Lambda_{11}}$ be the sampled part and let $C_{11}$ denote the contents of the corresponding cell.

Inductively, suppose that we have defined $\Lambda_{11},\ldots,\Lambda_{1k}$, cell contents $C_{11},\ldots,C_{1k}$, and set $A^{(11)},\ldots,A_{(1k)}$, so that all for points in $A^{(1k)}$, the query algorithm looks up $\Lambda_{(11)},\ldots,\Lambda_{(1k)}$ in the first $k$ lookups, given the contents $C_{11},\ldots,C_{1(k-1)}$. Inductively, we ensure that $\nu(A^{(1k)}) \leq \frac{1}{m^k}$. The $(k+1)$th lookup function, given the contents $C_{11},\ldots,C_{1k}$ partitions the set $A^{(1k)}$ into $m$ parts, and as above, we can refine this partition to get a $\frac{1}{m^{k+1}}$ sparse partitioning $A^{(1k)}_l : l \in [2m]$. We sample $\Lambda_{1(k+1)}$ uniformly from $[2m]$, and denote by $C_{1(k+1)}$ the contents of the corresponding cell. We define $A^{(1(k+1))}$ to be $A^{(1k)}_{\Lambda_{1(k+1)}}$. Clearly $A^{1(k+1)}$ has the desired inductive properties. Continuing in this fashion, we get $\Lambda_{11},\ldots,\Lambda_{1t}$, $C_{11},\ldots,C_{1t}$ and $A^{(11)},\ldots,A^{(1t)}$.

We repeat the above process $s$ times to get $s$ such paths. The matrix $\Lambda_{jk}$  ($j\in [s], k\in [t]$) denotes the sampled cell locations for the $s$ paths, the matrix $C_{jk}$ denotes the contents of these cells, and the sets $A^{(jk)}$ denotes the telescoping sequence of subsets for each of the $s$ paths.

For a technical reason, each entry in the first column of the $\Lambda$ matrix is drawn independently without replacement from $[2m]$, thus ensuring that this column consists of $s$ random distinct values from $[2m]$. For a matrix $U$ and sets $I,J$ of indices, let $U(I,J)$ denote the submatrix of $U$ indexed by $I$ and $J$.

\medskip\noindent{\em The measure of $\cup_{j\in[s]} A^{(jk)}$:} We first show that the measure $\nu(\cup_{j\in[s]} A^{(jk)})$ is concentrated around $\frac{s}{(2m)^k}$.
\begin{lemma}
\label{lem:numuconcentrated}
$\nu(\cup_{j\in[s]} A^{(jk)})$ is at most $(1+2k\sqrt{\frac{\log (t/\gamma)}{s}}) \frac{s}{(2m)^k}$, except with probability $\frac{k\gamma^2}{t^2}$.
\end{lemma}
\begin{proof}
We argue inductively. For $k=1$, the expected value of $Y_1\eqdef \nu(\cup_{j\in[s]} A^{(j1)})$ is exactly $\frac{s}{2m}$. Moreover, by $\frac{1}{m}$-sparsity of $A_1,\ldots,A_{2m}$, and using Chernoff bounds (for negatively correlated r.v.'s), the deviation from the mean is at most $\sqrt{2s\log (t/\gamma)}/m$, except with probability $\frac{\gamma^2}{t^2}$.

Inductively, the expected value of $Y_{k+1} \eqdef \nu(\cup_{j\in[s]} A^{(j(k+1))})$ is exactly $\frac{1}{2m}Y_k \leq (1+8k\sqrt{\frac{\log (t/\gamma)}{s}}) \frac{s}{(2m)^{k+1}}$  by the induction hypothesis. A Chernoff bound argument implies that the deviation from the mean is at most $\sqrt{2s\log(t/\gamma)}/m^{k+1}$, except with probability $\frac{\gamma^2}{t^2}$. The claim follows.
\end{proof}
	
For the rest of the proof, we will assume that for all $k\leq t$,  $\nuright(\cup_{j\in[s]} A^{(jk)})$ is indeed at most $(1 + 2k\sqrt{\frac{\log (t/\gamma)}{s}}) \frac{s}{(2m)^k} \leq \frac{3}{2} \frac{s}{(2m)^k}$ for as long at $t,w$ are $o(n^{\frac{1}{4}})$.

\medskip\noindent{\em The Encoder:} The encoding is simply set to be the matrix $C$. Note that the matrix $C$, along with $\Lambda$, which is part of the shared randomness, is enough to compute the answer computed by the data structure for every $y \in \cup_{j\in [s]} A^{(jt)}$.

In the rest of the proof, we argue that there is a good decoding algorithm that recovers most of the $b_i$'s. We do so in two steps. We first argue that if $x_i$ shatters at all levels, then the decoding algorithm succeeds with high probability. We then argue that in fact most $x_i$'s must shatter at all levels.

\medskip\noindent{\em If every $x_i$ shatters, the sampled cell contents are sufficient to estimate $b_i$'s:}

We first argue that the sampling process ensures that the majority vote on the sample agrees with the true majority, if the appropriate shattering happens at each level. For ease of notation, for $k \geq 1$, let $K_k = K m^{k-t}$. Note that by Observation~\ref{obs:Kbeta-transfer}, the $(K,\frac{1}{m^t})$-weak shattering implies $(K_k,\frac{1}{m^k})$-weak shattering. We start by defining the shattering event formally.

\begin{definition}[WShatter]
We will say that $WShatter_{k+1}(x,\Lambda([s],[k]), C([s],[k]))$ occurs if the collection $\{ A^{(jk)}_l: j\in[s],l\in[m]\}$ $(K_{k+1},\gamma^2/t)$-weakly shatters $x$. We use the notation $WShatter(x,\Lambda,C)$ to denote the event $\myand_{k \leq t} WShatter_{k}(x,\Lambda([s],[k]), C([s],[k]))$.
\end{definition}

When $\Lambda$ and $C$ are obvious from context, we will simply abbreviate these events as $WShatter_k(x)$ and $WShatter(x)$.

Based on Observation~\ref{obs:shatternut}, we define a sequence of measures. Let $\nut^0_x = \nu_x$. For any $k < t$, consider the collection $\{ A^{(jk)}_l: j\in[s],l\in[m]\}$. Let $\nup^{k+1}_x$ be the measure guaranteed by observation~\ref{obs:shatternut} so that $\nup^{k+1}_x$ is dominated by $\nu_x$, and satisfies $\nup^{k+1}_x(A^{(jk)}_l) \leq \frac{1}{K_{k+1}}$. Moreover, assuming $WShatter_{k+1}(x)$, $\nu_x(\cup_{j\in [s],l\in [2m]} A^{(jk)}_l)-\nup^{k+1}_x(\cup_{j\in [s],l\in [2m]} A^{(jk)}_l)$ is small. We set $\nut^{k+1}_x(y) = \min(\nut^{k}_x(y),\nup^{k+1}_x(y))$ for every $y \in V$. Thus  $\nut^{k+1}_x$ is a part of $\nut^k_x$ that is shattered at level $(k+1)$. We remark that if $x$ was strongly shattered at each level, $\nut^{k}_x = \nu_x$ for all $k$.

Let $V_0$ (resp. $V_1$) denote the set of vertices for which the query algorithm, given the table population, outputs $0$ (resp. $1$). So $V_b$ and $V_{b^c}$ are the set of queries for which the query algorithm outputs the bits $b$ and its complement $b^c$ respectively.

The decoding algorithm works as follows: for each $i$, we would like to compute the majority answer restricted to this set $\cup_{j\in [s]} A^{(jt)}$, under the measure $\nu_{x_i}$ restricted to this set. To deny any one random choice a large influence on the outcome, we take the majority under the measure $\nut^t_{x_i}$. Thus the decoder outputs
 $$\hat{b}_i = arg\,max_b \left( \sum_{j\in[s]}\nut^t_x(V_{b} \cap A^{(jk)})\right).$$

To prove that this decoding is usually correct, we show that the measure of $\cup_j A^{(jk)}$ under $\nut^{k}$ remains close to its expectation, and that the measure of points $y \in \cup_j A^{(jk)}$ where the data structure returns the wrong answer remains small. We define two more events.
\begin{definition}[Rep]
For $k\geq 1$, we let $Rep_{k}(x,\Lambda([s],[k]), C([s],[k-1]))$ be the event
$$
\sum_{j\in [s]} \nut^{k}_x(A^{(jk)}) \geq (1-\frac{3k\gamma}{t}) \frac{s}{(2m)^k}.
$$
\end{definition}

\begin{definition}[Small]
For $k\geq 1$, we let $Small_k(x,b,\Lambda([s],[k]), C([s],[k]))$ denote the event
$$
\sum_{j\in[s]}\nut^{k}_x(V_{b} \cap A^{(jk)}) \leq (2\gamma + \frac{k\gamma}{t})\frac{s}{(2m)^k}.
$$
For convenience, let $Small_0(x,b)$ denote the event $\nu_x(V_b) = \nut^0_x(V_b) \leq 2\gamma$.
\end{definition}

By the discussion above, except with probability $\gamma$, $Small_0(x_i,b_i^c)$ occurs for $(1-\gamma)$ of the $i$'s. We assume that this is indeed the case.

With the definitions in place, we are now ready to argue that each of these events happens for most $x_i$'s. For brevity, we use $Rep_k(x),Small_k(x,b)$ when the other arguments are obvious from context. It is immediate from the definitions that
\begin{lemma}
If $Rep_t(x_i)$ and $Small_t(x_i,b_i^c)$ occur, then the decoding $\hat{b_i}$ agrees with $b_i$.
\end{lemma}

We argue that assuming $WShatter_k(x)$ for each $k$, the events $Rep_t(x)$ and $Small_t(x)$ indeed happen with high probability. The following two lemmas form the base case, and the induction step of such an argument.

We first argue that
\begin{lemma}
\label{lem:pathbasecase}
For any $x,b$,
$$\Pr_{\Lambda([s],1)}[(Rep_{1}(x) \mid WShatter_{1}(x)] \geq 1-\frac{\gamma^2}{t}
$$
$$\Pr_{\Lambda([s],1)}[Small_{1}(x,b)) \mid Small_0(x,b)] \geq 1-\frac{\gamma^2}{t}
$$
\end{lemma}
\begin{proof}
Assuming $WShatter_1(x)$, we have
$$ \sum_{l=1}^{2m} \nut^1_x(A^{(10)}_l) \geq \sum_{l=1}^{2m} \nu_x(A^{(10)}_l) - \frac{\gamma^2}{t} \nu(\cup_{l=1}^{2m} A^{(10)}_l) \geq 1 - \gamma^2/t.$$

Since the $A^{(j1)}$  are drawn at random without replacement from these $2m$ sets,  each $A^{(10)}_l$ contributes $\frac{s}{2m}\nut^1_x(A^{(10)}_l)$ to the expectation of $Y \eqdef \sum_{j\in [s]} \nut^{1}_x(A^{(j1)})$ and this lower bounds $\Ex[Y]$. Finally, since these terms are negatively correlated, Chernoff bounds imply
$$ Y \geq (1-\frac{\gamma}{t})\Ex[Y],$$
except with probability $\exp(-\gamma^2 K_1\Ex[Y]/2t^2)$, since each term in $Y$ is in $(0,\frac{1}{K_1})$. The claim follows.

Next we argue $Small_1(x,b)$, i.e. we need to upper bound $\sum_{j\in[s]}\nut^{1}_{x}(V_{b} \cap A^{(jk)})$. The bound on the expectation follows from $Small_0(x,b)$ and the sampling. A Chernoff bound argument identical to that for $Rep_1$ completes the proof.
\end{proof}

Similarly, we argue  that
\begin{lemma}
\label{lem:pathindcase}
For any $k \geq 1$, any $x,b$,
$$\Pr_{\Lambda([s],{k+1})}[Rep_{k+1}(x) \mid (Rep_{k}(x) \myand WShatter_{k+1}(x))] \geq 1-\frac{\gamma^2}{t}
$$
$$\Pr_{\Lambda([s],{k+1})}[Small_{k+1}(x,b) \mid Small_k(x,b)] \geq 1-\frac{\gamma^2}{t}
$$
\end{lemma}
\begin{proof}

To prove $Rep_{k+1}(x)$, we use $WShatter_{k+1}(x)$ and $Rep_{k}(x)$. By weak shattering, $\sum_{j\in[s],l} \nut^{k+1}_x(A^{(jk)}_l) \geq \nut^{k}_x(\cup_{j\in[s]} A^{(jk)}) - \frac{\gamma^2}{t} \nu(\cup_{j\in[s]} A^{(jk)})$. By disjointness of the $A^{(jk)}$'s and using $Rep_k(x)$, the first term is at least $(1-\frac{3k\gamma^2}{t}) \frac{s}{(2m)^k}$. Using lemma~\ref{lem:numuconcentrated}, the second term is at most $\frac{\gamma^2}{t} \cdot \frac{3s}{2(2m)^k}$.
This lower bounds the expectation of $Y \eqdef \sum_{j\in [s]} \nut^{k+1}_x(A^{(j(k+1))})$ since each $A^{j(k+1)}$ is chosen uniformly from $A^{(jk)}_l$'s. The choices for different $j$'s are independent, so that Chernoff bounds imply that
$$ Y \geq (1-\frac{\gamma}{t})\Ex[Y],$$
except with probability $\exp(-\gamma^2 K_{k+1}\Ex[Y]/2t^2)$. A calculation identical to the previous lemma implies that $Rep_{k+1}(x)$ occurs in this case.

An identical Chernoff bound argument suffices to show $Small_{k+1}(x,b)$.
\end{proof}

And thus by induction,
\begin{lemma}
\label{lem:pathinduction}
For any $x,b$,
$$\Pr_{\Lambda([s],[t])}[Rep_{t}(x) \wedge \neg WShatter(x)] \geq 1-\gamma^2
$$
$$\Pr_{\Lambda([s],[t])}[Small_{t}(x,b) \mid Small_0(x,b)] \geq 1-\gamma^2
$$
\end{lemma}

If for most $x_i$, $WShatter(x_i)$ occurred, this would imply that the decoding algorithm succeeds with high probability. However, the event $WShatter_{k+1}(x_i)$ depends on the contents $C([s],[k])$, which are determined by the table population, which depends on $x_i$ itself.

\medskip\noindent{\em Proving that weak shattering happens for most $x_i$'s:}

The weak shattering property implies that for any $k$, for a fixed table population $T=T_1,\ldots,T_t$ (which determines $C([s],[k])$ via $\Lambda([s],[k])$),
$$\Pr_{x\sim \mu}[WShatter_{k+1}(x)] \geq 1-\frac{\gamma^2}{t},$$

Thus for any fixed $T$ and $\Lambda$, it is the case that
$$\Pr_{x\sim \mu}[WShatter(x)] \geq 1-\gamma^2.$$

Let $ManyShatter(T,\Lambda,\xvec)$ be the event that $WShatter(x)$ occurs for all but a $2\gamma$ fraction of the $x_i$'s in the instance $\xvec$, i.e. $\sum_{i} \one(WShatter(x_i,\Lambda([s],[k]),C(T,\Lambda)([s],[k]))) \geq  (1-2\gamma)n$.

\begin{lemma}
\label{lem:pathmanyshatter}
$$\Pr_{\xvec \sim \mu}[\exists T,\Lambda: \neg ManyShatter(T,\Lambda, \xvec)] \leq \exp(-\gamma^2 n/4))$$
\end{lemma}
\begin{proof}
First consider a fixed $T,\Lambda$. Since the $x_i$'s are drawn independently, Chernoff bounds imply that
$$\Pr_{\xvec \sim \mu}[\neg ManyShatter(T,\Lambda, \xvec)] \leq  \exp(-\gamma^2 n/2).$$

Further, note that the event $ManyShatter(T,\Lambda, \xvec)$ depends on $T$ only through the contents $C(T,\Lambda)([s],[t])$. Thus doing a union bound over all possible values of $C$ and $\Lambda$,
$$\Pr_{x_1,\ldots,x_n}[\exists T,\Lambda: \neg ManyShatter(T,\Lambda, \xvec)] \leq 2^{(st(w+\log 2m))} \exp(-\gamma^2 n/2).$$
The claim follows since $s=\frac{\gamma^2 n}{4t(w+\log 2m)}$.
\end{proof}

We assume for the rest of the proof that the database $\xvec = x_1,\ldots,x_n$ indeed has this property; this changes the failure probability by a negligible amount.

Let $T(\xvec,\bvec)$ be the table population built by the data structure. Lemma~\ref{lem:pathmanyshatter} implies that
$$\sum_i \one(WShatter(x_i)) \geq (1-2\gamma)n.$$

Thus using Lemma~\ref{lem:pathinduction},
$$\sum_i \one(Rep_t(x_i) \myand Small_t(x_i,b_i^c)) \geq (1-6\gamma)n.$$

Since $Dec(Enc(b,z),z)_i = b_i$ whenever $Rep_t(x_i) \myand Small_t(x_i)$, it follows that

$$\Ex[\sum_{i} \one(Dec(Enc(b,z),z)_i=b_i)]  \geq (1-6\gamma)n.$$

It follows that
$$\Ex[|Dec(Enc(b,z),z) - b|_1] \leq 6\gamma n.$$

The rare events ignored during the rest of the proof add an additional $O(\gamma n)$ to this expectation.

For small enough $\gamma$, the size of the encoding $stw$ is smaller than $(\gamma^2/4) n < rn$ (since $\lim_{\delta \rightarrow 0} r(\delta) = 1$), contradicting Corollary~\ref{cor:nocompression}. Hence the claim.
\end{proof}

\subsection{Cell Sampling}

In this section, we show a different sampling technique, which gives a different lower bound. The main theorem is:
\begin{theorem}
\label{thm:randpathsample}
There exists a constant $\gamma>0$ such that the following holds.
Let $(G,e)$ satisfy $\gamma$-weak independence (WI) and $(K,\frac{1}{m},\frac{\gamma}{t})$-weak shattering property. Then for $t \leq n^{\frac{1}{4}}$, any deterministic $t$-probe data structure for the distribution over GNS instances defined by $(G,e)$ that succeeds with probability $(1-\gamma)$ must use space $m$ at least $\Omega(\gamma^3nK^{\frac{1}{2t}}/wt^3\log (t/\gamma))$.
\end{theorem}
\begin{proof}
The proof is analogous to that for theorem~\ref{thm:randcellsample}. We define Encoding and Decoding procedures that compress a random string. The primary difference is in the sampling procedure; instead of sampling ``paths" as in the previous section, we sample cells from each table, and argue that the set of points that can be looked up only using paths in the sample is sufficient to recover the bits $b_i$. Specifically, we define events analogous to $Rep$ and $Small$ in the previous section, and show that they occur for many points.

We assume the contrary so that for $m < \gamma^3nK^{\frac{1}{2t}}/wt^3\log (t/\gamma))$, there is $t$-probe space $m$ data structure that succeeds with probability $1-\gamma^3$ on the distribution defined by $(G,\mu,\nucol)$. We use this data structure to construct functions $Enc$ and $Dec$ violating corollary~\ref{cor:nocompression}. We use the auxiliary input as shared randomness between $Enc$ and $Dec$ throughout this proof. The database $(x_1,b_1),\ldots, (x_n,b_n)$ is defined as before.

\medskip\noindent{\em Cell Sampling Procedure:}
Let $s=8t^2 m\log (t/\gamma)/\gamma^2K^{\frac{1}{2t}}$. Let $F_1 = F_1(y)$, $F_2 = F_2(y,\alpha_1)$, $F_3=F_3(y,\alpha_1,\alpha_2)$  etc. be adaptive lookup functions that the data structure uses. The sampling is done in $t$ steps one for each table. At each step we will get subsets $V=A^0 \supset A^1 \supset A^2, \ldots, A^t$ so that all the queries in $A^i$ only access the sampled cells for the first $i$ lookups.

Let $A^0= V$ and let $A^0_l$ for $l \in [2m]$ be a $\frac{1}{m}$-sparse partition obtained by refining $\{F_1^{-1}(l): l\in [m]\}$; this can be done as before by arbitrarily breaking up cells larger than $\frac{1}{m}$. Let $\Lambda_{11},\ldots,\Lambda_{s1}$ be a random subset of $[2m]$ of size $s$. Let the contents of the respective cells in a table population $T$ be denoted by $C_{11},\ldots,C_{s1}$. Thus the sample from the first table consists of $s$ rows $\Lambda_{11},\ldots,\Lambda_{s1} \in [2m]$ whose contents are $C_{11},\ldots C_{s1} \in \{0,1\}^w$. We set $A^1= \cup_{j\in [s]} A^0_{\Lambda_{j1}}$.

Let $A^1_l: l \in [2m]$ be a $\frac{1}{m}$-sparse partitioning obtained by refining $\{A^1 \cap F_2^{-1}(l): l\in [m]\}$ as above. We pick a random subset $\Lambda_{12},\ldots,\Lambda_{s2}$ of $[2m]$ of size $s$, and let $C_{12},\ldots,C_{s2}$ denote the relevant set of contents from $T$. We set $A^2 = \cup_{j\in [s]} A^1_{\Lambda_{j2}}$ be the set of queries $y \in V$ that look up one of the sampled cells in the first two tables.

Repeating this process, we get $s\times t$ matrices $\Lambda$ and $C$, and sets $A_1,\ldots,A_t$. Note that in any execution of the procedure, the set $A^k$ depends only on the samples $\Lambda$ and the contents $C$ read from the table population.

\medskip\noindent{\em The measure of $A^{k}$.} We first show that the measure $\nu(A^k)$ is concentrated around $(\frac{s}{2m})^k$.
\begin{lemma}
\label{lem:numuconcentratedcell}
$\nu(\cup_{j\in[s]} A^{k})$ is  at most  $(1 + 2k\sqrt{\frac{\log (t/\gamma)}{s}}) (\frac{s}{2m})^k$, except with probability $\frac{k\gamma^2}{t^2}$.
\end{lemma}
\begin{proof}
We argue inductively. For $k=1$, the expected value of $Y_1\eqdef \nu(A^1)$ is exactly $\frac{s}{2m}$. Moreover, by $\frac{1}{m}$-sparsity of $A_1,\ldots,A_{2m}$, and using Chernoff bounds (for negatively correlated r.v.'s), the deviation from the mean is at most $\sqrt{2s\log (t/\gamma)}/m$, except with probability $\frac{\gamma^2}{t^2}$.

Inductively, the expected value of $Y_{k+1} \eqdef \nu(A^{k+1})$ is exactly $\frac{s}{2m}Y_k \leq (1+ 2k\sqrt{\frac{\log (t/\gamma)}{s}}) (\frac{s}{2m})^{k+1}$ by the induction hypothesis. A Chernoff bound argument identical to the one above completes the proof.
\end{proof}
	
For the rest of the proof, we will assume that for all $k\leq t$,  $\nuright(A^k)$ is indeed at most $(1+ 2k\sqrt{\frac{\log (t/\gamma)}{s}}) (\frac{s}{2m})^k \leq \frac{3}{2} (\frac{s}{2m})^k$.

\medskip\noindent{\em The Encoder:} The encoding is set as before to be the matrix $C$. Note that the matrix $C$, along with $\Lambda$, which is part of the shared randomness, is enough to compute the answer computed by the data structure for every $y \in A^t$.

\medskip\noindent{\em If every $x_i$ shatters, the sampled cell contents are sufficient to estimate $b_i$'s:}

We argue that the sampling process ensures that the majority vote on the sample agrees with the true majority, if the appropriate shattering happens at each level. We start by defining the shattering event formally.

\begin{definition}[WShatter]
We will say that $WShatter_{k+1}(x,\Lambda([s],[k]), C([s],[k]))$ occurs if the collection $\{ A^k_l: l\in[m]\}$ $(K,\gamma^2/t)$-weakly shatters $x$. We use the notation $WShatter(x,\Lambda,C)$ to denote the event $\myand_{k \leq t} WShatter_{k}(x,\Lambda([s],[k]), C([s],[k]))$.
\end{definition}

When $\Lambda$ and $C$ are obvious from context, we will simply abbreviate these events as $WShatter_k(x)$ and $WShatter(x)$.

We once again define a sequence of measures. Let $\nut^0_x = \nu_x$. For any $k < t$, consider the collection $\{ A^k_l,l\in[m]\}$. Let $\nup^{k+1}_x$ be the measure guaranteed by observation~\ref{obs:shatternut} so that $\nup^{k+1}_x$ is dominated by $\nu_x$, and satisfies $\nup^{k+1}_x(A^k_l) \leq \frac{1}{K}$. Moreover, assuming $WShatter_{k+1}(x)$, $\nu_x(\cup_l A^k_l)-\nup^{k+1}_x(\cup_l A^k_l)$ is small. We set $\nut^{k+1}_x(y) = \min(\nut^{k}_x(y),\nup^{k+1}_x(y))$. We remark that if $x$ was strongly shattered at each level, $\nut^{k}_x = \nu_x$ for all $k$.

Let $V_0$ (resp. $V_1$) denote the set of vertices for which the query algorithm, given the table population, outputs $0$ (resp. $1$). So $V_b$ and $V_{b^c}$ are the set of queries for which the query algorithm outputs the bits $b$ and its complement $b^c$ respectively.

The decoding algorithm works as follows: for each $i$, we would like to compute the majority answer restricted to this set $A^{t}$, under the measure $\nu_{x_i}$ restricted to this set. Instead we take the majority under the measure $\nut^t_{x_i}$. Thus the decoder outputs
 $$\hat{b}_i = arg\,max_b \left(\nut^t_x(V_{b} \cap A^{t})\right).$$

To argue that the decoding is correct for most $x_i$'s, we show that the measure of $A^{k}$ under $\nut^{k}$ remains close to its expectation, and that the measure of points $y \in A^{k}$ where the data structure returns the wrong answer remains small. We define two more events.
\begin{definition}[Rep]
For $k\geq 1$, we let $Rep_{k}(x,\Lambda([s],[k]), C([s],[k-1]))$ be the event
$$
\nut^{k}_x(A^{k}) \geq (1-\frac{3k\gamma}{t}) (\frac{s}{2m})^k.
$$
\end{definition}

\begin{definition}[Small]
For $k\geq 1$, we let $Small_k(x,b,\Lambda([s],[k]), C([s],[k]))$ denote the event
$$
\nut^{k}_x(V_{b} \cap A^{k}) \leq (2\gamma + \frac{k\gamma}{t})(\frac{s}{2m})^k.
$$
For convenience, let $Small_0(x,b)$ denote the event $\nu_x(V_b) = \nut^0_x(V_b) \leq 2\gamma$.
\end{definition}

As in the path sampling case, except with probability $\gamma$, $Small_0(x_i,b_i^c)$ occurs for $(1-\gamma)$ of the $i$'s. We assume that this is indeed the case.

With the definitions in place, we are now ready to argue that each of these events happens for most $x_i$'s. For brevity, we use $Rep_k(x),Small_k(x,b)$ when the other arguments are obvious from context. It is immediate from the definitions that
\begin{lemma}
If $Rep_t(x_i)$ and $Small_t(x_i,b_i^c)$ occur, then the decoding $\hat{b_i}$ agrees with $b_i$.
\end{lemma}

We argue that assuming $WShatter_k(x)$ for each $k$, the events $Rep_t(x)$ and $Small_t(x)$ indeed happen with high probability. The following two lemmas form the base case, and the induction step of such an argument.

We first argue that
\begin{lemma}
\label{lem:cellbasecase}
For any $x,b$,
$$\Pr_{\Lambda([s],1)}[(Rep_{1}(x) \mid WShatter_{1}(x)] \geq 1-\frac{\gamma^2}{t}
$$
$$\Pr_{\Lambda([s],1)}[Small_{1}(x,b)) \mid Small_0(x,b)] \geq 1-\frac{\gamma^2}{t}
$$
\end{lemma}
\begin{proof}
Assuming $WShatter_1(x)$, we have
$$ \sum_{l=1}^{2m} \nut^1_x(A^{0}_l) \geq \sum_{l=1}^{2m} \nu_x(A^{0}_l) - \frac{\gamma^2}{t} \nu(\cup_{l=1}^{2m} A^{0}_l) \geq 1 - \gamma^2/t.$$

Since we draw $s$ sets at random without replacement from these $2m$ sets,  each $A^{0}_l$ contributes $\frac{s}{2m}\nut^1_x(A^{0}_l)$ to the expectation of $Y \eqdef \nut^{1}_x(A^1) = \sum_{j\in [s]} \nut^{1}_x(A^{0}_{\Lambda_{j1}})$ and this lower bounds $\Ex[Y]$. Finally, since these terms are negatively correlated, Chernoff bounds imply
$$ Y \geq (1-\frac{\gamma}{t})\Ex[Y],$$
except with probability $\exp(-\frac{\gamma^2}{2t^2} K(\frac{s}{2m}))$. The claim follows by an easy calculation.

Next we argue $Small_1(x,b)$, i.e. we need to upper bound $\sum_{j\in[s]}\nut^{1}_{x}(V_{b} \cap A^{1})$. The bound on the expectation follows from $Small_0(x,b)$ and the sampling. A Chernoff bound argument identical to that for $Rep_1$ completes the proof.
\end{proof}

Similarly, we argue  that
\begin{lemma}
\label{lem:cellindcase}
For any $k \geq 1$, any $x,b$,
$$\Pr_{\Lambda([s],{k+1})}[Rep_{k+1}(x) \mid (Rep_{k}(x) \myand WShatter_{k+1}(x))] \geq 1-\frac{\gamma^2}{t}
$$
$$\Pr_{\Lambda([s],{k+1})}[Small_{k+1}(x,b) \mid Small_k(x,b)] \geq 1-\frac{\gamma^2}{t}
$$
\end{lemma}
\begin{proof}
To prove $Rep_{k+1}(x)$, we use $WShatter_{k+1}(x)$ and $Rep_{k}(x)$. By weak shattering, $\sum_{l\in [2m]} \nut^{k+1}_x(A^{k}_l) \geq \nut^{k}_x(A^{k}) - \frac{\gamma^2}{t} \nu(A^{k})$. By $Rep_k(x)$, the first term is at least $(1-\frac{3k\gamma^2}{t}) (\frac{s}{2m})^k$. Using lemma~\ref{lem:numuconcentratedcell}, the second term is at most $\frac{\gamma^2}{t} \cdot \frac{3}{2}(\frac{s}{2m})^k$.
This lower bounds the expectation of $Y \eqdef \nut^{k+1}_x(A^{(k+1)}) = \sum_{j\in [s]} \nut^{k+1}_x(A^{k}_{\Lambda_{j(k+1)}})$ since $\Lambda_{k+1}$ is a uniformly random subset of $[2m]$ of size $s$. Chernoff bounds then imply that
$$ Y \geq (1-\frac{\gamma}{t})\Ex[Y],$$
except with probability $\exp(-\frac{\gamma^2}{2t^2} K\Ex[Y])$. A calculation similar to the previous lemma implies that $Rep_{k+1}(x)$ occurs in this case.

An identical Chernoff bound argument suffices to show $Small_{k+1}(x,b)$.
\end{proof}

And thus by induction,
\begin{lemma}
\label{lem:cellinduction}
For any $x,b$,
$$\Pr_{\Lambda([s],[t])}[Rep_{t}(x) \wedge \neg WShatter(x)] \geq 1-\gamma^2
$$
$$\Pr_{\Lambda([s],[t])}[Small_{t}(x,b) \mid Small_0(x,b)] \geq 1-\gamma^2
$$
\end{lemma}

If for most $x_i$, for all $k$, $WShatter_{k}(x)$ occurred, this would imply that the decoding algorithm succeeds with high probability.

The rest of the proof is identical to that for Theorem~\ref{thm:randcellsample}, since once again, the event $WShatter(x)$ depends on the table population only through the contents $C$.

\end{proof}

\section{Applications}
\label{sec:applications}

We show how lower bounds on GNS imply lower bounds for ANNS.
We stress that these bounds hold for the \emph{average case} where the $n$ data-set points are sampled randomly from a distribution over $|V|$. Thus, if with high probability the distance between all pairs of  points in the data set is at least $cr$, then the bounds above hold also for the approximate nearest neighbor within factor $c$.
The following table lists all these bounds and how they follow from our work.

\begin{table}[h]
    \centering

        \begin{tabular}{c|r|c|r|r|r}
        \textbf{metric space}  & \textbf{appr} & \textbf{det / rand} & \textbf{bound } & \textbf{ref} & \textbf{Thm}\\  \hline

        $\ell_1$  & $1/\e$  & det & $t \geq d\e^2/\log(mwd/n)$ &\cite{PatrascuT06a}, \cite{Liu04} & \ref{thm:deterministic}\footnote{Our result improves the dependence on $\e$ slightly from $\e^3$ to $\e^2$} \\ \hline
        $\ell_1$  & $1+\e$  & rand & $mw \geq n^{\tfrac{1}{\e^2 t}}$&\cite{AIP06} & \ref{thm:randomized} \\ \hline
        $\ell_1$  & $1/\e$  & rand & $mw \geq n^{1+\tfrac{\e}{t}}$ &\cite{PTW08} & \ref{thm:randomized} \\ \hline
        $\ell_{\infty}$  & $\log_\rho\log d$  & det & $mw \geq n^{\rho/t}$ &\cite{ACP08} & \ref{thm:deterministic} \\ \hline
        \end{tabular}
        \label{tab:sum}
    \caption{Known lower bounds, and how they follow from Theorems~\ref{thm:deterministic} and \ref{thm:randomized}}
\end{table}

\subsection{GNS to ANNS}

In the decisional version of the $(c,r)$-ANNS problem we have a metric space $\mathcal M$ and parameters $c$ and $r$. We preprocess $n$ points $x_1,...,x_n$ into a data structure. When given a query point $y$ the goal is to distinguish between the case where $d(x_i,y) \leq r$ for some $i \in [n]$, and the case where for all $i$ $d(x_i,y)\geq cr$. The query algorithm is required to output $1$ in the former case, $0$ in the latter case, and may report anything if neither of the two cases hold.

We show that if we have appropriate distributions over $\mathcal M$, we can derive lower bounds for $(c,r)$-ANNS by simply computing the relevant expansion parameter.

\begin{theorem}[GNS to ANNS Deterministic]
Let $c \geq 1$ and let $\mu$ be a distribution over a metric ${\mathcal M}= (V,d)$ satisfying:
$$(c\mbox{-Strong Independence})\;\;\;\;\;\;\;\;\Pr_{x,z \sim \mu} [d(x,z) \leq (c+1)r] \leq \frac{1}{100n^2}.$$
Let $G_r = (V, \{(u,v): d(u,v) \leq r\})$. Then GNS on $(G_r,\mu)$ reduces to $(c,r)$-approximate GNS on $\mathcal M$.
\end{theorem}
\begin{proof}
Given a GNS instance $(x_1,b_1),\ldots,(x_n,b_n)$, we consider the dataset $D_1=\{x_i: b_i = 1\}$ as our input for the $(c,r)$-ANNS problem. It is easy to see that when $x_i \sim \mu$ and $b_i \sim \{0,1\}$, this set $D_1$ is a uniformly random dataset from $\mu$. The $c$-Strong independence implies strong independence for the GNS instance. Whenever $d(x_i,x_j) > (c+1)r$, and $d(x_i,y) \leq r$, we have $d(x_j,y) > cr$ so that for the  $c$-approximate NNS instance, the answer is $1$ if and only $x_i$ is in $D_1$, i.e. if and only if $b_i=1$. The claim follows.
\end{proof}

Thus to prove deterministic data structure lower bounds for $c$-approximate NNS, it suffices to exhibit $r$ and a distribution $\mu$ which satisfies $c$-strong independence and has large expansion.

Similarly
\begin{theorem}[GNS to ANNS Randomized]
Let $c \geq 1$ and let $e$ be a distribution over pairs of points in a metric ${\mathcal M}= (V,d)$. Let $\mu(x) = e(x,V)$ and $\nu(y) = e(V,y)$. Suppose that for small enough $\gamma$
$$(c\mbox{-Weak Independence})\;\;\;\;\;\;\;\;\Pr_{y \sim \nu, z \sim \mu} [d(y,z) \leq cr] \leq \frac{\gamma}{n},$$
and
$$ \Pr_{(x,y) \sim e}[d(x,y) \leq r] \geq 1-\gamma.$$
Then GNS on $(G,e)$ reduces to $(c,r)$-approximate GNS on $\mathcal M$.
\end{theorem}
\begin{proof}
As before, given a GNS instance $(x_1,b_1),\ldots,(x_n,b_n)$, we consider the dataset $D_1=\{x_i: b_i = 1\}$ as our input for the $(c,r)$-ANNS problem. It is easy to see that when $x_i \sim \mu$ and $b_i \sim \{0,1\}$, this set $D_1$ is a uniformly random dataset from $\mu$. The properties above imply weak independence for the GNS instance. Finally, except with small probability, we have $d(x_j,y) > cr$ for all $j \neq i$, so that for the  $c$-approximate NNS instance, the answer is $1$ if and only $x_i$ is in $D_1$, i.e. if and only if $b_i=1$. The claim follows.
\end{proof}

Thus to prove randomized data structure lower bounds for $c$-approximate NNS, it suffices to exhibit $r$ and a distribution $e$ which satisfies the above properties and has large expansion.

\subsection{Computing Expansion}
\label{sec:hypercubeexp}
Next we bound the expansion of the hypercube for appropriate distributions, which would imply the claimed lower bounds for the Hypercube. The vertex expansion result for $\ell_\infty$ in~\cite{ACP08} will lead to the cell probe lower bounds for $\ell_{\infty}$ proved in their work.

We will set $\mu$ to be the uniform distribution over the hypercube. For a set $A \subseteq H$, we let $a = \mu(A)$.
Observe that if we take $n$ uniformly random points from a $d$ dimensional hypercube then with high probability all pairs of points are at least $d/2 - O(\sqrt {d \log n})$ apart. Using bounds for the expansion for the ($G_r$ corresponding to the) $d$-dimensional hypercube, we will derive lower bounds for the near neighbor problem on the hypercube.

The following lemma is proved in~\cite{Bollobas86}
\begin{lemma}[Vertex Expansion of Hypercube]
Let $H=\{0,1\}^d$ be the boolean hypercube, and let $G_r$ be the graph with the edge set $E_r=\{(u,v): |u-v|_1 \leq r\}$.
Let $h_i = \frac{1}{2^d}\sum_{j=0}^i {d \choose j}$.
Then the vertex expansion $\Phi_v(h_i) \geq {h_{i+r}}$
\end{lemma}

Setting $r = d/3$ we see that the node expansion $\Phi_v \ge h_{d/2}/h_{d/2 - d/3} = 2^{\Omega(d)}$. From theorem~\ref{thm:deterministic} we get $(mwt/n)^2 \ge 2^{\Omega(d)}$ or $t \geq d/log(mwd/n)$. This gives us the deterministic lower bound for $2-$approximation in $\ell_1$ norm. Setting $r = \epsilon d/2$ gives us an a lower bound for $O(\epsilon)-$ approximation. For this value of $r$, $\Phi_v \ge h_{d/2}/h_{(1-\epsilon)d/2} = 2^{\Omega(\epsilon^2 d)}$. So we get the bound of $t \geq \epsilon^2 d/log(mwd/n)$.

For randomized lower bounds, we will use the distribution defined by the noise operator $T_{\rho}$ where $\rho = (1-\frac{2r}{d})$. I.e. to sample from $e$, we sample $x$ from the uniform distribution $\mu$ and sample $y$ by flipping each bit of $x$ independently with probability $\frac{(1-\rho)}{2}=\frac{r}{d}$. It is easy to check that for any $r \leq (\frac{1}{2}-\e)d$ and $d$ being $\Omega(\log n/\eps^2)$, we indeed have $\Pr_{(x,y) \sim e}[d(x,y) \leq (1+\frac{\e}{10})r] \geq 1- \frac{\gamma}{n}$. Moreover, since $\mu=\nu$, by the discussion above, $\Pr_{y\in \nu,z\in \mu}[d(y,z) \leq \frac{d}{2} - O(\sqrt {d \log n})] \leq \frac{\gamma}{n^2}$. Thus it remains to compute the expansion for appropriate $r$.

It will be convenient to work with the edge expansion.
\begin{definition}
We define the edge expansion $\Phi_e(\delta)$ for a $(G,e)$ as $\Phi_e (\delta) = min_{\mu(A) \leq \delta} \frac{e(A,V)}{e(A,A)}$.  Thus for any set of size measure $\delta$, at most $\frac{1}{\Phi_e(\delta)}$ mass of edges incident on $A$ stay within $A$.
\end{definition}

\begin{observation}
For any $(G,e)$, if $\mu$ is uniform, then $\Phi_{r}(\delta, \gamma) = \Omega(\gamma\Phi_e (2\delta))$
\end{observation}
\begin{proof}
First we will argue that for any sets $A$ and $B$ where $\mu(A) = \mu(B)\leq \delta$, $e(A,B)\leq 2 e(A,V) /\Phi_e(2\delta)$. To see this note that $e(A,B) \leq e(A \cup B, A \cup B) \leq \frac{1}{\Phi_e(\mu(A\cup B))} e(A \cup B, V) =
\frac{1}{\Phi_e(2\delta)} 2 e(A,V)$

Now consider any set $A$ of measure at most $\delta$, and let $B$ by any other set of measure $\delta\gamma\Phi_e(2\delta)/2$. We wish to argue that $e(A,B) \leq \gamma e(A,V)$ which would imply the claim.

Let $B_1,\ldots,B_k$ be a partition of $B$ into $k = \lceil\frac{\delta\gamma\Phi_e(2\delta)}{2\delta}\rceil$ pieces of measure $\delta$ each. $e(A,B) \leq \sum_i e(A,B_i) \leq \frac{k}{\frac{1}{\Phi_e(\mu(A\cup B))}} \leq \gamma$. The claim follows.

\end{proof}

\begin{lemma}[Edge expansion of Hypercube]
\label{lem:hypercube-edge}
Let $H=\{0,1\}^d$ be the boolean hypercube, and $(G,e)$ be as above for $r < \frac d 4$. Then the edge expansion $\Phi_e(a) \geq a^{-\Omega(r/d)}$.
\end{lemma}
\begin{proof}
For sets $A,B$, it is easy to see that $e(A,B) =  \langle T_{\rho} \one_A,  \one_B\rangle$. But by the Hypercontractive inequality,
\begin{align*}
\langle T_{\rho} \one_A, \one_A\rangle = \langle T_{\sqrt{\rho}} \one_A, T_{\sqrt{\rho}}\one_A\rangle = \| T_{\sqrt{\rho}} \one_A\|_2^2  \leq |1_A|_{1+\rho}^2 = a^{\frac{2}{1+\rho}}.
\end{align*}
Also $e(A,V) = a$. The claim follows by substituting the value of $\rho$.
\end{proof}

Setting $r = \epsilon d/2$, we get that for constant $\gamma$, $\Phi_{r}(a,\gamma) \geq a^{-\Omega(\epsilon)}$. From theorem ~\ref{thm:randomized} (first inequality) it follows that $(mw/n)^t \geq m^{-\Omega(\epsilon)}$ implying $m \geq (n/w)^{1+\Omega(\epsilon/t)}$.

\begin{lemma}[Robust expansion of Hypercube]
Let $H=\{0,1\}^d$ be the boolean hypercube, and let $(G,e)$ be as above for $r=(\frac{1}{2}-\e)d$ Then for any sets $A$ and $B$ such that $\mu(B) \leq \mu(A)^{4(1-\frac{2r}{d})^2}$, $e(A,B) \leq \mu(A)^{(1-\frac{2r}{d})^2}d e(A,V)$.
\end{lemma}
\begin{proof}
As in the proof of lemma~\ref{lem:hypercube-edge}, we use the hypercontractive inequality.
\begin{eqnarray*}
e(A,B) &=& \langle T_{\rho} \one_A, \one_B\rangle\\
  &\leq &\|T_{\rho} \one_A\|_2 \|\one_B\|_2\\
&\leq& \|\one_A\|_{1+\rho^2} \|\one_B\|_2\\
&=& a^{\frac{1}{1+\rho^2}} b^{\frac 1 2}\\
&\leq& a^{1-\rho^2} a^{2\rho^2}\\
&=& a^{\rho^2} e(A,V).
\end{eqnarray*}
The claim follows.
\end{proof}
\begin{corollary}
For any $\beta \geq 0$, and $r \leq \frac{d}{2}$, setting $\rho=1-\frac{2r}{d}$, $\Phi_{av}(\beta,\beta^{\rho^2}) \geq \beta^{1-4\rho^2}$ for $G_r$ as above.
\end{corollary}
Setting $r = (1-\eps) d/2$, we see that from theorem~\ref{thm:randomized} that any randomized algorithm must use space
$m \geq (\frac{n}{wt^5})^{\frac{4}{t\epsilon^2}}$.

The vertex expansion result for $\ell_\infty$ has already been computed in the work by ~\cite{ACP08} for proving cell probe lower bounds for $\ell_\infty$.
They consider the $d$-dimensional grid $\{0,1,\ldots,m\}^d$ with a non uniform measure defined as follows. The measure $\pi$ over $\{0,1,\ldots,m\}$ is defined by  $\pi(i) = 2^{-(2\rho)^i}$ for all $i>0$ and $\pi(0) = 1 - \sum_{i>0} \pi(i)$. One then defines $\mu_d(x_1,x_2,\ldots, x_d) = \pi(x_1)\cdot\pi(x_2)\ldots \pi(x_d)$. For this measure $\mu_d$ the following expansion theorem was shown in ~\cite{ACP08}.

\begin{lemma}~\cite{ACP08}
$\mu_d(N(A)) \geq \mu_d(A)^{1/\rho}$
\end{lemma}

Now it is easy to show that if $n$ random points are chosen from the measure $\mu_d$ then every point is at least
$g =\log_{2\rho} (\frac{\epsilon}{4} \log d)$ away from the origin under the $\ell_\infty$ norm. This is because the probability that a certain coordinate of a point is at least $\Omega(g)$ is at least $1/d^{\epsilon}$. So probability that no coordinate is more than $g$ is at most $(1-1/d^{\epsilon/4})^d \leq e^{-d^{1-\epsilon/4}}$. For $d = \Omega(\log^{1+\epsilon} n)$ this is at most $1/n^{\Omega(1)}$. So all points will be at least distance $g$ from the origin. In fact this argument also easily shows that they are in $g$ distance from each other. So setting $r=1$ gives us a lower bound for $(g-1,r)-$ANNS for $\ell_\infty$. From Theorem~\ref{thm:deterministic} it follows that to get a $O(\log_{\rho}\log d)$ approximation for NNS on $\ell_\infty$ the amount of space required is at least $m \geq (\frac{n}{wt})^{\rho/t}$.

\section{A Matching Upper Bound}
\label{sec:ubound}
We already know that the lower bound is tight in many specific cases. Here we show that the tightness holds more generally, for highly symmetric graphs. For such graphs, we show that the notion of robust expansion correctly captures the complexity of GNS for the regime with a constant number of queries.

Let $G$ be an undirected Cayley graph\footnote{We need the graph to be highly symmetric and the symmetries of Cayley graphs are convenient for the claims we need.}, and assume that it has the weak independence property for the uniform distribution. Let $m$ be such that $m=n\Phi_{r}(\tfrac{1}{m},\gamma)$ (denoted by $\Phi$ for brevity) where $\gamma \geq \tfrac{3}{4}$. Below we describe a data structure with $m$ cells and word size $O(\log n)$ which can solve GNS in a single query with constant probability for the hard distribution of inputs. This matches the lower bound of Theorem \ref{thm:randomized} for the case $t=1$.

\begin{theorem}
Let $G$ be an undirected Cayley graph that has the weak independence property for the uniform distribution. Let $m$ be such that $m=n\Phi_{r}(\tfrac{1}{m},\gamma)$ (denoted by $\Phi$ for brevity) where $\gamma \geq \tfrac{3}{4}$. Then there is a $1$-probe data structure that uses $m$ words of $w=\log |G|$ bits each that succeeds with constant probability when the data-set points $x_1,\ldots,x_n$ are drawn randomly and independently, and the query point is a random neighbor of a random $x_i$.
\end{theorem}
We observe that this is the distribution for which we show the lower bound.
\begin{proof}
The main idea is to use the low expanding sets in order to construct something similar to a  Locality Sensitive Hashing solution.  We stress that the upper bound is in the cell probe model, which allows us to ignore the (practically very important) issue of actually computing the LSH efficiently.

Let $A\subset V$ be a set of measure $1/m$ for which the robust expansion is  $\Phi$ and $m=n\Phi$. By the definition of robust expansion, we know that there is a set $B$ of measure $\Phi/m$ such that $|E(A,B)| \geq \gamma |E(A,V)|$. We take $m$ random translations of $A$ and $B$, denoted by $A_1,\ldots,A_m$ and $B_1,\ldots,B_m$, formally, we sample uniformly $m$ elements $a_1,...,a_m$ from the group underlying the Cayley graph, and set $A_i = \{u:   u= a_i + v, v\in A \}$ and similarly $B_i = \{u:   u= a_i + v, v\in B \}$. The translation by $a_i$ is an automorphism that maps $A$ to $A_i$ and $B$ to $B_i$ so for each $i$ $|E(A_i,B_i)|\geq \gamma |E(A_i,V)|$.

We construct a table $T$ with $m$ cells as follows: Given a data set point $x$, we check for each $i\leq m$ whether $x \in B_i$, and if so we place $x$ in $T_i$. Note that the measure of each $B_i$ is $\Phi/m$ so that the expected number of data set points $x_j$ that fall in $B_i$ is $n\Phi/m$ which is $1$. For random data sets, most $B_i$'s will contain at most (say) $10$ data-set points. In order to keep the word size small we store at most $10$ data set points in each table cell $T_i$, and assuming that $O(\log n)$ bits suffice to represent a data-set point, we have $w = O(\log n)$.

Now, given a query point $y$ we find an $i$ for which $y \in A_i$ and output the data-set point in $T[i]$ which is closest to $y$. Note that with constant probability, such an $i$ exists and is unique.

Recall that the data set $x_1,\ldots,x_n$ is obtained by sampling $n$ points uniformly and independently from $V$. Further, we assume that this distribution is weakly independent; i.e. if $x$ and $y$ are random nodes and $z$ is a random neighbor of $x$, then $\Pr[z \in N(y)] \leq 1/100n$. Further, the query point is obtained by sampling a random neighbor of a random data set point. Assume that the correct answer is $x$. Now if $y\in A_i$ (which happens with a constant probability), then the lookup succeeds if $x \in B_i$ and there were less than $10$ data set points in $B_i$. The first event occurs with probability $\gamma \geq \tfrac{3}{4}$ and the second event occurs independently with probability at least $\tfrac{3}{4}$ as well.  We conclude that the data structure succeeds with constant probability.
\end{proof}

\section{Low Contention Dynamic Data Structures}
\label{sec:dynamic}
Let $Q_l$ denote the set of queries that read cell $T[l]$ from the table.
\begin{definition} A data structure is said to have contention $\tau$ if $\nu(Q_l)\leq \tau$ for all $l$.
\end{definition}
Say we insert a new point $x$. After the insertion there would be a subset of $N'(x) \subseteq N(x)$ such that the query algorithm outputs the correct answer for any $y \in N'(x)$.   We say the insertion $x$ is successful if the measure of this set under $\nu_x$ is at least $1-\gamma$.

We remark that we look at contention only for the part of the data structure that depends on the database; accesses to any randomness are free.

We  prove Theorem~\ref{thm:dynamic}, which states that in such data structures, the update time is at least $\Omega(\Phi_r(\tau,O(\frac{1}{t^2}))/t^4)$.

\begin{proof}[Proof of Theorem \ref{thm:dynamic}]

Consider the state of a $t$-probe data structure just before we insert a point. The $t$ lookup functions each give a partitioning of $V$ such that each part in the partitioning has measure at most $\tau$ under $\nu$. Let $Q^i_1,\ldots,Q^i_m$ be the $i$th partitioning.

Let $L$ be the set of locations in $T$ that are updated when one inserts $x$. Thus $|L| \leq t_U$. Further, note that for the answer for $y$ to be correct both before and after the insertion, $L$ must intersect with at least one of the locations that are queried on $y$, i.e. $y \in Q^i_{L}$ for some $i$.

By assumption each of the $i$ partitions is $\tau$-sparse. Lemma~\ref{lem:exptoshatweak} implies that except with small probability, $x$ is $(K,\gamma/2t)$ shattered where $K=\Phi_r(\tau,\frac{\gamma^2}{4t^2})\gamma^3/16t^3$ by each of the partitions. Strong shattering would imply that for each $l$, the measure $\nu_x(Q^i_l) \leq \frac{1}{K}$ so that $|L|$ changes can account for at most a $t|L|/K$ measure being affected, which would imply the result.

Since we only have weak shattering, we recall that shattering implies that for any $i$, there a measure $\nut^i_x$ such that $\nut^i_x(V) \geq (1-\frac{\gamma}{2t})$ and $\nut^i_x(Q^i_l) \leq \frac{1}{K}$ for any $l$. Let $\nut_x(y) \eqdef \min_i \nut^i_x(y)$. Thus $\nut_x(V) \geq (1-\frac{\gamma}{2})$.

Finally,
\begin{eqnarray*}
1-\gamma &\leq& \nu_x(\cup_i Q^i_L)\\
&\leq& \frac{\gamma}{2} + \nut_x(\cup_i Q^i_L)\\
&\leq& \frac{\gamma}{2} + \sum_{i=1}^t \nut_x(Q^i_L)\\
&\leq& \frac{\gamma}{2} + t|L|/K
\end{eqnarray*}
Thus $t_U \geq |L| \geq K/2t = \Phi_r(\tau,\frac{\gamma^2}{4t^2})\gamma^3/32t^4 $.
\end{proof}

\end{document}